\documentclass[reqno,11pt]{amsart}

\usepackage{amsmath}
\usepackage{amssymb}
\usepackage{amsthm}

\usepackage{ifthen}
\usepackage{graphicx}

\usepackage{psfrag}

\usepackage[utf8]{inputenc}
\usepackage[T1]{fontenc}

\numberwithin{equation}{section}

\numberwithin{equation}{section}

\newtheorem{thm}{Theorem}[section]
\newtheorem{lm}[thm]{Lemma}
\newtheorem*{thm*}{Theorem}
\newtheorem{corollaire}[thm]{Corollary}
\newtheorem*{corollaire*}{Corollary}

\newtheorem{prop}[thm]{Proposition}
\newtheorem*{prop*}{Proposition}
\newtheorem*{definition}{Definition}
\newtheorem*{num_res}{Numerical result}
\newtheorem*{num_ress}{Numerical results}
\newtheorem*{conj}{Conjecture}

\newcommand{\R}{\mathbb{R}}

\newcommand{\cC}{{\ensuremath{\mathcal C}} }

\newcommand{\cO}{{\ensuremath{\mathcal O}} }

\newcommand{\bbS}{{\ensuremath{\mathbb S}} }

\DeclareMathOperator{\vspan}{Span}

\begin{document}

\title[Asymptotic dynamics for the Kuramoto model]{
Global attractor and asymptotic dynamics in the Kuramoto model for coupled noisy phase oscillators}

\author{Giambattista Giacomin}
\address{
  Universit{\'e} Paris Diderot (Paris 7) and Laboratoire de Probabilit{\'e}s et Mod\`eles Al\'eatoires (CNRS),
U.F.R.                Math\'ematiques, Case 7012 (site Chevaleret)
             75205 Paris Cedex 13, France
}
\author{Khashayar Pakdaman}
\address{Institut Jacques Monod, CNRS, UMR 7592, Univ Paris Diderot,
Sorbonne Paris Cit\'e, F-750205 Paris, France}

\author{Xavier Pellegrin}
\address{Institut Jacques Monod, CNRS, UMR 7592, Univ Paris Diderot,
Sorbonne Paris Cit\'e, F-750205 Paris, France}

\date{
21th July 2011, 
revised version 5th January 2012}

\begin{abstract}
We study the dynamics of the large $N$ limit of the Kuramoto model of coupled phase oscillators, 
subject to white noise. We introduce the notion of shadow inertial manifold and we prove their existence for this model, 
supporting the fact that the long term dynamics of this model is finite dimensional. 
Following this, we prove that the global attractor of this model takes one of two forms. 
When coupling strength is below a critical value, the global attractor is a single equilibrium point corresponding to an incoherent state.
 Otherwise,  when coupling strength is beyond this critical value, the global attractor is a two-dimensional disk 
composed of radial trajectories connecting a saddle-point equilibrium (the incoherent state) to an invariant closed curve 
of locally stable equilibria (partially synchronized state). Our analysis hinges, on the one hand, 
upon sharp existence and uniqueness results and their consequence for the existence of a global attractor, and, on the other hand, 
on the study of the dynamics in the vicinity of the incoherent and coherent (or synchronized) equilibria. We prove in particular non-linear stability 
of each synchronized equilibrium, and normal hyperbolicity of the set of such equilibria. We explore mathematically and 
numerically several properties of the global attractor, in particular we discuss the limit of this attractor as noise intensity decreases to zero.
\end{abstract}

\maketitle

\section{Introduction}

The Kuramoto model is a classical model of interacting phase oscillators, coupled through a mean-field term, widely used to investigate synchronization. 
It has been applied in various fields such as physics, chemistry and biology (see review \cite{RevModPhys.77.137} and reference therein).
 In different contexts the same model is also referred to as the Sakaguchi model \cite{MR1783382},
 the Smoluchowski model \cite{MR2164412,MR2109485,MR2425333}, or the mean-field classical spin XY model \cite{silver:468}.
 One remarkable feature of this model is that, for large populations, it undergoes a transition from an incoherent state to a 
(partially) synchronous one as the coupling strength between oscillators is increased. 
Many numerical and theoretical studies have analyzed this transition under various hypotheses \cite{MR1783382}.
 In the present work, we are interested in the long term dynamics of infinitely many interacting identical noisy phase oscillators. 
For such an infinite population, the following mean field model can be derived from the microscopic description (see \cite{MR2594897} and references therein)
\begin{equation}\label{eq}
\begin{cases}
\partial_{t} q(t, \theta) = \frac{1}{2} \partial_{\theta}^2q(t, \theta) - 
   K \partial_{\theta} \left [ q(t,\theta) (J \ast q)(t, \theta) \right ] & \; \; t > 0,\,  \theta \in [-\pi, \pi],  \\
q(0, \theta) = q_{0}(\theta) & \; \; \theta \in [- \pi , \pi], \\
q(t, - \pi) = q(t, \pi) & \; \; t \geq 0, \\
\partial_{\theta} q(t, - \pi) = \partial_{\theta} q(t, \pi) & \; \; t \geq 0.    
\end{cases}
\end{equation}
where
\begin{equation}\label{eq:pot_kuramoto}
 J \ast q (t,\theta) = \int_{-\pi}^{\pi} \sin( \varphi - \theta) q(t,\varphi) d\varphi,
\end{equation} $K$ is a real constant representing the coupling strength $q \geq 0$ such that $\int_{\mathbb{S}} q (t,\theta) d \theta=1$. 
The variable $\theta \in [- \pi, \pi]$ accounts for the phase of the oscillators, and the unknown $q(t,\theta)$ 
for the density of oscillators at phase $\theta$ at time $t$. 
We have assumed that oscillators are homogeneous, i.e. they all have the same intrinsic frequency $\omega=0$, 
and without loss of generality we have assumed that
the intensity of the noise perturbing each oscillator is one.
\par For this system one encounters two distinct regimes depending on the ratio between the noise intensity and 
the coupling strength. When the coupling strength $K$ is smaller than a critical value $K_c$ the noise dominates, 
a uniform state is the only equilibrium of (\ref{eq}), and the population always tends to this incoherent state. 
When $K>K_c$ instead the coupling dominates, a family of non-trivial coherent (or synchronized) equilibria exists, and the population tends to synchronize.
%
%
When $K > K_c$, for a different coupling potential $\tilde J \ast q (t,\theta) = \int_{-\pi}^\pi \sin^2(\varphi - \theta) q(\theta) d\theta $, 
existence and uniqueness of a one dimensional circle  $\mathcal C = \{ f^*(\cdot + \varphi), \; \varphi \in [0,2\pi[ \}$, where 
$f^*$ is a non uniform probability on $[0, 2\pi]$,
 of non-trivial equilibria has been established in \cite{LuoZZ_2004}.
In a more general setting, estimates on the number of equilibria, and asymptotics in the large coupling limit $K \rightarrow + \infty$, 
have been established in \cite{CKT_2005}.
For equation (\ref{eq}) with (\ref{eq:pot_kuramoto}), when $K>K_c$, existence and uniqueness of a circle  of non-trivial equilibria 
$\mathcal C $ has been established in \cite{MR2109485} (see also \cite{MR2594897} and references therein).
A spectral gap estimate and linear stability of trivial equilibria $\hat q(\cdot + \varphi)$ have also been shown in \cite{MR2594897}.

\par Several major advances have recently been made in the understanding of the global dynamics of Kuramoto models, 
in the case of deterministic oscillators and in the case of noisy oscillators. 
For an infinite number of identical oscillators subject to no noise, equivalent to equation (\ref{eq}) with no diffusion, 
Ott and Antonsen have shown thanks to a well chosen ansatz
that a two-dimensional invariant manifold exists, and they have provided an explicit expression for the orbits on 
this manifold \cite{MR2464324}. These orbits are the paths followed by the oscillator population during synchronization, 
as they describe the trajectories that connect the incoherent state to synchronized ones. 
Generalizing these ideas to finite dimensional systems, Mirollo and Strogatz \cite{MR2603661} 
have shown that regardless of the (finite) number of oscillators, the dynamics of 
the homogeneous Kuramoto model can be reduced to three dimension, through an action of 
a three-dimensional Moebius group on the torus $\mathbb{T}^{N}$. 
While these constitute important breakthroughs in our understanding 
of the dynamics of coupled phase oscillators, they do not deal with the situation where noise perturbs the dynamics of the units.

\par In the case of an infinite number of noisy identical oscillators (\ref{eq}), 
Vukadinovic has established in \cite{MR2425333} 
 the existence of finite dimensional invariant exponential attractors (inertial manifolds) 
 in the invariant subspace of symmetric solutions $q(t,\theta) = q(t, - \theta)$ of (\ref{eq}). 
Inertial manifolds existence theorems do not apply directly to equations of this form, but 
Vukadinovic has developed methods to show their existence for a Smoluchowski equation 
on the circle \cite{MR2425333} and on the sphere \cite{MR2470912}, 
and for a Burgers equation \cite{MR2725294} for example. 
Although the dynamics of (\ref{eq}) is infinite dimensional, 
inertial manifolds show that it has some typical properties of finite dimensional systems. 
\par In this work we are interested in infinite populations of noisy coupled oscillators
(\ref{eq}) and the finite-dimensional behavior of its dynamics. This work is in line with previous findings and it extends them 
to give a complete rigorous description of the long term dynamics of the Kuramoto model (\ref{eq}).
This paper is organized so as to move progressively from general results concerning the solutions of equation (\ref{eq}) to more refined 
descriptions of the dynamics on its attractor and ending with a number of numerical investigations opening the way for some conjectures.
\par In section \ref{section_existence}, we start with sharp existence and regularity results for solutions of (\ref{eq}).
Existence and uniqueness of solutions for (\ref{eq}) in $L^2$ or Sobolev $H^s$ spaces is a 
classical result. Combining a result of \cite{MR2594897} and a method of \cite{MR2164412}, we show that
for any initial condition $q_0$ in a measure class the (unique) solution $q(t)$ of (\ref{eq}) is in a analytical functions space for all times $t >0$.
One consequence of the regularizing properties of equation (\ref{eq}) is that many convergence phenomena of solutions towards equilibria 
or invariant manifolds happen in an analytic-functions space and not just in the classical $L^2$ space. 
\par In section \ref{section_inertial}, we prove the existence of shadow inertial manifolds for the Kuramoto model (\ref{eq}). 
As inertial manifolds which have been introduced to overcome some flaws of global attractors \cite{MR943945} \cite{MR943276}, 
 shadow inertial manifolds attract all solutions exponentially fast, the dynamics on a shadow inertial manifold is given by an ODE system, 
and each solution of the system has a phase on an asymptotically complete shadow inertial manifold.
 In that respect 
 inertial manifolds and shadow inertial manifold give finite dimensional and accurate reduction of the long term behavior of a dynamical system. 
Our proof uses of original ideas of  \cite{MR2425333}, while avoiding some of its technical difficulties. 

\par These previous results establish essentially that the long term dynamics of equation (\ref{eq}) can be captured by finite dimensional ODEs.
From this point on, we focus on more specific properties of its  asymptotic dynamics. In section \ref{section_equilibria}, 
we analyze the stability (or lack of stability) of equilibria. 
For readers' convenience, we start by briefly recalling existing results for equilibria of equation (\ref{eq}) of \cite{MR2109485,MR2594897} mentioned above. 
Following this we perform local stability analyzes of equilibria. 
We write linearization theorems at the incoherent equilibrium $q(\theta)=\frac{1}{2\pi}$
in a rigorous mathematical setting, and we show that they confirm what has been largely expected in the literature.  
This thorough analysis also gives global stability information and estimates of escape times when the incoherent equilibrium is unstable,  
 and it will be a basic ingredient in the study on the global attractor in section \ref{section_global}.
At the synchronized equilibria, an ad hoc Hilbert structure and associated spectral gap estimates have recently been found \cite{MR2594897}. 
Thanks to these we show that, when the coupling strength is large enough, $K>K_c$, the family of synchronized equilibria is asymptotically stable. 
Due to rotational invariance of (\ref{eq}), this family forms a circle like closed invariant curve. 
We give an estimate of the phase-shifting effect of a small perturbation of a synchronized equilibrium on this invariant curve. 
Finally, another important consequence of these local stability results is the uniform normal hyperbolicity 
of the family of synchronized equilibria that is of interest on its own. More on this can be found in \cite{GPPP}.

\par In section \ref{section_global}, we establish several properties of the global attractor of (\ref{eq}).
The existence of a global attractor in a classical consequence of regularizing properties of equation (\ref{eq}), 
such as proved in \cite{MR2594897,MR2164412,MR1691574}. This is briefly recalled in section 
\ref{section_global1}.
In section \ref{section_global2}, relying on (un)stability analysis of equilibria in section \ref{section_equilibria}, 
we show that, just as without the diffusion term, equation (\ref{eq}) possesses a two-dimensional invariant manifold. 
Interestingly, this invariant manifold coincides with the global attractor of the model, 
whereas in the case of equation (\ref{eq}) without diffusion, the global attractor, if it exists, 
cannot be reduced to the two dimensional invariant manifold of \cite{MR2464324}:
there are many unstable equilibria in the no-diffusion case 
that the global attractor would have to contain. 
Global attractors are 
classical tools to obtain a finite dimensional reduction of the asymptotic behavior of a large system \cite{MR941371}, 
this means that Kuramoto model (\ref{eq}) can asymptotically be reduced to a simple two-dimensional dynamical system on a disk. 
We characterize the trajectories on this manifold as orbits connecting the incoherent state 
to a synchronized one and further provide an algorithm for their numerical computation. 
Several properties of the global attractor of (\ref{eq}), the dynamics on it and in its neighborhoods are also discussed 
numerically in section \ref{section_global3}. 
In the strong-coupling limit
(or equivalently in the small-diffusion limit ) $K \rightarrow + \infty$, we show that the global attractor of (\ref{eq}) converges formally 
to the two-dimensional invariant manifold found by Ott and Antonsen \cite{MR2464324} in the no-diffusion case of equation (\ref{eq})
 and supported by numerical evidences, we conjecture that this convergence holds in analytical functions spaces. 
\par In this way, we derive a full description at different levels of the long term dynamics of the mean field coupled Kuramoto model composed of
infinitely many identical phase oscillators in the presence of noise.

\section{Existence, uniqueness and regularity results}\label{section_existence}


\par 
In the following we consider Sobolev spaces 
 $ H^s(\mathbb S) = \left  \{ q \in L^2(\mathbb S): \, \sum_{1}^{+ \infty} (1+k^2)^s ( \alpha_k^2 + \beta_k^2 ) < \infty  \right \} $
for $s \geq 0$, and the Gevrey spaces
 $ \mathcal G_{a} = \left   \{ q \in L^2(\mathbb S): \, 
 \Vert q \Vert_{\mathcal G_a}^2 = \sum_{k=1}^{+ \infty} a^k (\alpha_{k}^{2} + \beta_{k}^{2}) < + \infty \right \}. $
When $a>1$, the Gevrey space $\mathcal G_a$ is a subspace on the space of real analytic functions, 
for any $k$ the natural injections $\mathcal G_a \rightarrow H^k$ are continuous and compact, 
and if $1< a_1 < a_2$ then the injection $ \mathcal G_{a_2} \rightarrow \mathcal G_{a_1}$ is continuous and compact. 
Notice that in a different context Gevrey functions are usually between $C^\infty$ regularity and analytic regularity, 
but for partial differential equations the Gevrey spaces $\mathcal G_a$ above are classical too \cite{MR1026858}.

If $ q(t, \theta) = \frac{1}{2 \pi} + \frac{1}{\pi} \sum_{n=1}^{+ \infty} x_n(t) \cos(n \theta) + y_n(t) \sin(n \theta) $ 
is a solution of (\ref{eq}) in $C^{1}([0, + \infty[, L^{2}(\mathbb{S}))$, taking $z_{n} = x_{n} + i y_{n}$ one finds that for all $n \geq 1$ and $t \geq 0$,
\begin{equation}\label{eq_FourierODEs}
 z_{n}^{\prime} = - \frac{n^{2}}{2} z_{n} + \frac{Kn}{2}[ z_{1} z_{n-1} - \overline{z_{1}} z_{n+1}  ]\, , 
\end{equation}
where $z_{0}(t) = 1$ for all times. 
This ODE system is the key point to show that solutions are in a Gevrey space for any positive time, 
details of computations are similar to those of theorem 3.1 in \cite{MR2164412} for the symmetric case $y_n=0$ ($\forall n \geq 0$ and $t \geq 0$).
In \cite{MR2594897}, it was shown that for $q_0$ in a measure class, the solution $q(t)$ of 
(\ref{eq}) is $C^\infty(]0, +\infty[ \times \mathbb S^1)$. 
Combining these, we can show that Gevrey regularization occurs in the general case 
($y_n \neq 0$) too, and regardless of the coupling strength $K$
the solution of Kuramoto equation is in $\mathcal G_a$ for any $a > 1$ (for all $t>0$).
This proof also inspires the following lemma \ref{Gevrey_Cv_Lemma}, establishing that equation (\ref{eq}) defines a continuous semiflow 
$L^2 \rightarrow  \mathcal G_a$, which will be used several times in the next sections. 

\begin{thm}
\label{th:G1}
Let $q_{0}$ be a probability measure on $\mathbb{S}$. 
There is a unique weak solution of the Kuramoto equation $q_{t} \in C^{0}([0 , + \infty[, \mathcal{M}(\mathbb{S}))$, 
absolutely continuous with respect to the Lebesgue measure for any $t>0$, 
and denoting $q(t, \cdot)$ its density we have 
$q(t, \theta) > 0$ for all $ t >0$ and $\theta \in [0, 2 \pi]$, we have $ q \in C^{\infty}( ]0 , + \infty[\times [0 , 2 \pi ] )$ and
 $q(t, \cdot) \in \mathcal G_{a}$ for any constant $a \geq 1$ and $t>0$.
Furthermore, for any $a>1$ there is a bounded absorbing set $B_a$ in Gevrey space $\mathcal G_a$ for equation (\ref{eq}).
\end{thm}

\par From now on we will denote $S_t$ the semiflow associated to equation (\ref{eq}).
A bounded subset $B_a$ is an absorbing set for a semiflow $S_t$, defined in a metric space $X$, if for any bounded set $B \subset X $ 
the trajectories initiated in $B$
enter $B_a$ in finite time and remain in that set thereafter, ie $\exists t_0$ $\forall x_0 \in B $ $\forall t \geq t_{0}$, $x(t) \in B_a$.
In theorem \ref{th:G1} above, $B_a$ is a compact absorbing set for $S_t$ in $X=H^s$ for any $s \in \mathbb N$ 
and in $X= \mathcal G_{a^\prime}$ for any $1< a^\prime < a$.

\par The following lemma will be used several times in the following sections to show convergence
results in Gevrey spaces. Its proof can be found in the Appendix.
\begin{lm}(Gevrey-convergence Lemma)
\label{Gevrey_Cv_Lemma} 
For any positive time $\epsilon>0$ and any constant $a \geq 1$, the semiflow $S_{\epsilon} : L^{2} \rightarrow \mathcal G_a $ 
associated to equation (\ref{eq}) is Lipschitz-continuous on bounded sets.
In particular, if $q$ and $\tilde{q}$ are two solutions of equation (\ref{eq}) bounded in $L^2$, 
then for any $\epsilon >0$, there is a Lipschitz constant $L$ (depending only on $\epsilon$, $\Vert q_0 \Vert_{L^2}$ and $\Vert \tilde q_0 \Vert_{L^2}$)
 such that for any time $t \geq 0$, 
\begin{equation} 
\Vert q(t+\epsilon) - \tilde{q}(t+\epsilon) \Vert_{\mathcal G_a} \leq L \Vert q(t) - \tilde{q}(t) \Vert_{L^2}.
\end{equation}
\end{lm}

\section{Asymptotically complete shadow inertial manifolds}
\label{section_inertial}
Inertial manifolds are finite dimensional, Lipschitz and positively
invariant manifolds that exponentially attract solutions of evolution
equations.
There is an extensive literature on the existence of inertial
manifolds and related invariant sets such as inertial sets for partial
differential equations (\cite{MR943945,MR1873467,MR1409653,MR1050131,MR1335230}). However, these general results do
not apply directly to nonlinear Fokker-Planck equations similar to the
Kuramoto model.
Indeed, existence theorems for inertial manifolds require a spectral
gap condition on the spectrum of the linearized equation that FPEs
such as the Kuramoto model do not satisfy. This hurdle
notwithstanding, Vukadinovic made a significant progress by
establishing the existence of inertial manifolds for a number of non
linear FPEs including the Kuramoto model (\cite{MR2425333,MR2470912,MR2725294}). For the latter,
his result is stated in the case of even solutions. Our aim in this
section is to show that a structure reminiscent of inertial manifolds
exists for the Kuramoto model whether solutions are even or not. As
some of the key steps of our proof are 
similar to Vukadinovic's work, we briefly review his strategy. This short review
also clarifies why we modify his approach in order to analyze the
situation where initial data are not necessarily even functions.

\par For a linear Fokker-Planck equation of the form
\begin{equation}\label{eq:FkP}
\partial_{t} q = \partial_{\theta}^2q + \partial_{\theta}( q
\partial_{\theta}(V) ),
\end{equation}
a standard method
to obtain a regular reaction term is to set $u = e^Vq$, which gives
\begin{equation}\label{eq:FkP_u}
\partial_t u = \partial_{\theta}^2 u + F[V]u.
\end{equation}
\par Kuramoto equation (\ref{eq}) can be written in the form of
(\ref{eq:FkP}), but is not linear because the ``potential''
$V(\theta) = - K \int_{-\pi}^\pi \cos(\theta - \varphi) q(\varphi)
d\varphi $ depends on $q$.
Nevertheless, in \cite{MR2425333} Vukadinovic has shown that in the
invariant subspace of even solutions of (\ref{eq}),
the map $q \mapsto e^V q$ is a bi-Lipschitz bijection onto its image,
and so that equation (\ref{eq:FkP_u}) can be written in
closed form in $u$. The reaction term $u \mapsto F[V(u)]u$ is then
well defined $H \rightarrow H$ for the appropriate Hilbert space $H$,
and standard theorems for the existence of inertial manifolds hold for
(\ref{eq:FkP_u}). Using the inverse transform $u = e^V q \mapsto q$,
Vukadinovic was able to establish the first inertial manifold
existence results for equations of the form (\ref{eq:FkP}). Proving
that the transform $q \mapsto u= e^{V}q$ is one to one, establishing
its regularity and the regularity of its inverse are the most
difficult and specific parts of \cite{MR2425333}.

This method cannot be readily extended to the non-symmetric case,
because it is not clear whether the map $q \mapsto V(q)$ is one-to-one
for non-even $q \in L^2$. This motivates us to modify the approach to
circumvent this question by embedding the system in a larger one where
it is possible to prove the existence of inertial manifolds. This is
sketched below with proofs left to the appendices. Prior to that we
introduce the following definition of asymptotically complete shadow
inertial manifold (AcSIM) and comment on it.
\begin{definition}
Suppose that there exists a compact absorbing set $B_a$ in $X$ for the
semiflow $S_t$.
Let $\mathcal M$ be a (finite Hausdorff dimensional) subset of $X$.
The set $\mathcal M$ an asymptotically complete shadow inertial
manifold (AcSIM) for $S_t$ on $X$ if
\begin{itemize}
\item there exists a smooth function $\Phi : \mathbb R^n \rightarrow X
$ such that $\mathcal M = \Phi(\mathbb R^n)$ is the graph of $\Phi$,
\item $\mathcal M$ attracts exponentially all trajectories of $S_t$:
there is a $\delta >0$
\begin{equation}
d(S_t x_0, \mathcal M) = \mathcal O \left ( e^{- \delta t} \right )
\;\;\; \forall x_0 \in X,
\end{equation}
\item there is a flow $\tilde S_t$ on $\mathbb R^n$ (associated to an
ODE system)
such that
 and for any $x_0 \in X$, there is a unique phase $\phi_0 \in
\mathbb R^n$ such that
$ d(S_t x_0, \Phi ( \tilde S_t \phi_0 ) ) = \mathcal O \left ( e^{- \delta t}
\right ). $
\end{itemize}
\end{definition}
An AcSIM gives a finite dimensional reduction of the dynamics of $S_t$, to which the trajectories converge exponentially fast, 
and from a practical point of view this reduction to $\mathcal M$ gives as much information as an inertial manifold would do. 
The flow $\tilde S_t$ is a shadow inertial form.
An AcSIM is an (asymptotically complete) inertial set as soon as it is positively invariant by $S_t$.
An AcSIM is a $C^1$ manifold as soon as the function $\Phi$ is $C^1$ and injective,  
and if it is both a manifold and an inertial set, an AcSIM is a classical AcIM.  
The uniqueness of the phase $\phi_0$ implies that the trajectories $P \tilde S_t$ on $\mathcal M$ capture the slow dynamics of $S_t$ on $X$ 
(dynamics at speed less than $\mathcal O \left ( e^{- \delta t} \right )$). 
When they exist, inertial manifolds, inertial sets or approximate inertial manifolds contain the global attractor of a dynamical system. 
The same type of result holds for asymptotically complete shadow inertial manifolds. 

\par We introduce now several notations and we sketch the main steps
of proof of the existence of AcSIMs for the Kuramoto model.
Naturally, several of these steps are similar to \cite{MR2425333},
notably when it comes to checking that the transformed system admits
an inertial manifold.  We consider the injection
$ I : \,q \mapsto \left (  q , x_1 , y_1  \right )$ where $x_1 =
\int_{-\pi}^\pi q(\theta) \cos(\theta) d\theta$
 and $y_1 = \int_{-\pi}^\pi q(\theta) \sin(\theta) d\theta$, and we
denote the projection $P: \, (q,a,b) \mapsto q$.
We define the transform
$\mathcal F : \, \left ( \begin{array}{c} q \\ a \\ b \end{array}
\right ) \mapsto
 \mathcal U = \left ( \begin{array}{c} e^Vq \\ a \\ b \end{array} \right )$
where $V = - \frac{K}{2} \left( a \cos(\theta) + b \sin(\theta) \right)$.
We will also denote $u = e^Vq$.
We will show that if $q$ is a solution of equation (\ref{eq}), then
$\mathcal U = \mathcal F \circ I(q)$
is solution of a reaction-diffusion problem
\begin{equation}\label{eq2} \frac{d \mathcal{U}}{dt} + A \mathcal{U} =
N(\mathcal{U}),  \end{equation}
where $N$ is well defined $H^s \times \mathbb R^2 \rightarrow H^s
\times \mathbb R^2$ for any $s \geq 0$.
We will then check that equation (\ref{eq2}) is well posed, that it
has a bounded absorbing set,
and that the eigenvalues of $A$ satisfy the spectral gap condition
ensuring the existence of inertial manifolds
 $\mathcal M^*$ for equation (\ref{eq2}).
We will denote by $S^*_t$ the semiflow generated by equation
(\ref{eq2}) in the following.
Then, we will show that the transform $\mathcal F$ is regular and bijective,
and we will deduce the existence of an AcSIM for equation (\ref{eq})
from the existence of $\mathcal M^*$.
More precisely, we will show the existence of a finite dimensional
$C^1$ graph $\mathcal M$ which attracts exponentially
any solution of equation (\ref{eq}).
For any initial condition $q_0 \in H^s$, the corresponding solution
$S_t q_0$ of (\ref{eq})
has (for any large enough $t_0$) a unique phase $v_0 \in \mathcal
M^*$, characterized by
\begin{equation}
\Vert S_{t+t_0} q_0 - P \mathcal F^{-1}\tilde S_t v_0 \Vert_{H^s} =
\mathcal O (e^{- \eta t }), 
\end{equation}
where $\tilde S_t$ is the restriction of $S^*_t$ to (the finite dimensional manifold) $\mathcal M^*$.
This graph $\mathcal M$ is an asymptotically complete shadow inertial manifold.
\par This strategy makes use of a new kind of attractors, i.e. AcSIM.
It does not rely on the injectivity of the potential $q \mapsto V$,
which makes it simpler.
It improves the functions spaces in which exponential convergence is
proved from $L^2$ to $H^s$ for any $s \geq 0$, and
it allows the definition of inertial form $\tilde S_t = S^*_t
\vert_{\mathcal M^*}$ (the ODE system ruling the dynamics on $\mathcal
M^*$),
which allow numerical finite-dimensional approximation of the dynamics
of (\ref{eq}) in principle.
We will also see that the transform $q \mapsto \left ( e^V q , x_1 ,
y_1  \right )$
is $C^k : \, H^s \rightarrow H^s \times \mathbb R^2 $ so that results
of Rosa and Temam \cite{MR1409653} may be applied to equation
(\ref{eq2}),
and show that the inertial manifolds $\mathcal M^*$ (and the dynamics
on it) are $C^k$ and uniformly normally hyperbolic.
The proof of this theorem is detailed in the appendix.

\begin{thm}\label{thm_PIM}
For any $s \in \mathbb N$, the set $\mathcal M = P \mathcal F^{-1}
\left ( \mathcal M^* \right )$ is a $C^1$ finite dimensional graph in
$H^s$, there is a $\eta >0 $
and there is a flow $\tilde S_t = S_t^*\vert_{\mathcal M^*}$ such that
for any $q_0 \in H^s$ and $t_0$ large enough,
 there is a unique $ v_0 \in \mathcal M^* $ such that
\begin{equation}
\Vert S_{t+ t_0} q_0 - P \mathcal F^{-1} \tilde S_t v_0 \Vert_{H^s} =
\mathcal O \left ( e^{-\eta t} \right ).
\end{equation}
We have in particular $dist_{H^s}( S_{t+t_0} q_0, \mathcal M )
 \leq \Vert S_{t+ t_0} q_0 - P \mathcal F^{-1} \tilde S_t v_0
\Vert_{H^s} = \mathcal O \left ( e^{-\eta t} \right ) $.
The graph $\mathcal M$ is exponentially attractive for the flow $S_t$
generated by equation (\ref{eq}), $\tilde S_t v_0$ is the phase
associated to the solution
$S_t q_0$ of (\ref{eq}), and $P \mathcal F^{-1} \tilde S_t v_0$ is its
shadow on $\mathcal M$. 
\par Furthermore, the set of equilibria of equation (\ref{eq2})
is the image of the set of equilibria of (\ref{eq}) by the transform
$\mathcal F \circ I$,
and the global attractor of equation (\ref{eq2}) is the image of the
global attractor of (\ref{eq}) by the transform $\mathcal F \circ I$.
\end{thm}

\section{Equilibria of the system : Stability analysis}
\label{section_equilibria}

\subsection{On the stationary solutions to \eqref{eq}}
\label{sec:equilibria}
Up to a rotation all stationary solutions  of \eqref{eq} which are probability densities --
we call them {\sl equilibria} -- can be written as 
\begin{equation} 
\label{eq:qhat}
\hat{q}(\theta) = c e^{2Kr\cos(\theta)} \, ,
\end{equation} 
with $r$ a solution of the equation 
\begin{equation}
\label{eq:fixedpoint}
 r = \frac{I_1 (2Kr)}{I_0(2Kr)}, \ \ \text{ with } I_j(s)= \frac1{2\pi}
 \int_0^{2\pi } (\cos(\theta))^j \exp(s \cos(\theta)) d \theta\ \ \
 \text{ for } s\in \R\, ,
 \end{equation}
and $c$ such that $\int_{0}^{2\pi} \hat{q}(\theta) d \theta = 1$.
We refer to \cite{MR2109485,MR2425333} for this, and to \cite{MR2594897} for this and a review of the mathematical physics literature on this issue,
 and a number of side facts.
See \cite{LuoZZ_2004} for similar results in the case of equation (\ref{eq}) with 
$\tilde J \ast q (t,\theta) = \int_{-\pi}^{\pi} \sin^2(\theta-\varphi)q(\varphi) d\varphi $, 
and see \cite{CKT_2005} for asymptotic estimates when $K \rightarrow + \infty$ in a three dimensional case similar to (\ref{eq}).
\par Equations (\ref{eq:qhat}) and (\ref{eq:fixedpoint}) directly imply that the fixed point equation \eqref{eq:fixedpoint}
can have at most three solutions: 
$r=0$, which is always present, and  gives $\hat{q} (\cdot)= \frac{1}{2 \pi}$: this corresponds to the incoherent state of the system;
and two non trivial solutions $\pm r$ (we denote by $r$ the positive solution) that exist if and only if when $K>1$. 
If we set $ \hat{q}_{\varphi}(\theta) := \hat{q}(\theta + \varphi)\, ,$
where $\hat q$ is given in \eqref{eq:qhat} with $r$ the nontrivial 
positive solution of \eqref{eq:fixedpoint}, then $\hat{q}_{\varphi}$
is an equilibrium for every choice of $\varphi$, so that we have a 
{\sl circle} $\cC= \{\hat{q}_\varphi: \, \varphi \in \bbS \} $ of equilibria
where each $q_{\varphi}$ describes the (partially) synchronized state around the phase $\varphi$.
\par Equation (\ref{eq}) can be written in a gradient-flow form
$ \partial_{t} q = \partial_{\theta} \left(  q \partial_{\theta} \frac{ \delta \mathcal{F} }{ \delta q( \theta )} \right),$
where the functional $\mathcal F $ is not increasing along solutions of equation (\ref{eq}).
The existence of a Lyapunov functional $\mathcal{F}$ gives information about the asymptotic behavior of the system, 
and the stability of equilibria, see for example  \cite{MR1420187,MR1957346}.
\medskip

\par In the following sections we will go beyond these results, by exploiting
 properties of the system linearized at the equilibria and by using invariant manifold theorems.
Additionally, the following stability and unstability results will also 
be useful for the study of the global attractor (section \ref{section_global}).
 We do this separately for the incoherent stationary solution and for
the synchronized solutions.

\subsection{The incoherent equilibrium $\frac{1}{2 \pi}$}

When $K \leq 1$, the constant $\frac{1}{2 \pi}$ is the only stationary solution, and the existence 
of a Lyapunov functional implies that it is globally attractive. 
\medskip 

\begin{prop}
\label{Prop_Stabilite_K<=1} 
For any value of $K$ and any initial condition $q_0 \in \mathcal{M}(\mathbb{S}) $, the $\omega$-limit set 
$\omega(q_0) = \underset{T \geq 0}{\cap} \overline{ \{S_t q_0, t \geq T \}}$ is included in the set of equilibria of equation (\ref{eq}).
In particular if $K \leq1$, for any $q_{0} \in \mathcal{M}(\mathbb{S})$ we have 
$\left  \Vert q(t) - \frac{1}{2\pi} \right \Vert_{\mathcal G_a}  \underset{t \rightarrow + \infty}{\longrightarrow} 0. $
\end{prop}
\begin{proof}
The convergence to $\frac{1}{2\pi}$ is a consequence of both the existence of a Lyapunov functional, and the fact that $\frac{1}{2\pi}$ 
is the unique equilibrium of (\ref{eq}). 
Lemma \ref{Gevrey_Cv_Lemma} implies that this convergence happens in Gevrey spaces $\mathcal G_a$.
\end{proof}

\medskip 

We now use invariant manifolds to describe precisely the dynamics around $\frac{1}{2\pi}$ and to have estimates of the speed 
of convergence to $\frac{1}{2\pi}$.
\begin{prop}
\label{Prop_Stabilite_K<1}
When $K<1$, there are constants $\epsilon > 0 $ and $M < \infty$ such that if 
$\Vert q_{0} - \frac{1}{2 \pi} \Vert_{H^{1}} \leq \epsilon$ then the solution of equation (\ref{eq}) satisfies 
\begin{equation} 
\label{eq:PS<1}
\left\Vert q(t) - \frac{1}{2 \pi} \right\Vert_{\mathcal G_a} \leq M  \left\Vert q_{0} - \frac{1}{2 \pi} \right\Vert_{H^{1}} \, e^{-\frac{1-K}{2}t} 
\ \ \  \text{ for every } t \geq 0. 
\end{equation}
Moreover we have 
\begin{equation} 
 q(t) = \frac{1}{2 \pi} + P(q_{0})e^{- \frac{1-K}{2} t} + \epsilon(t) \, , 
\end{equation}
with $\Vert \epsilon(t) \Vert_{H^{1}} \leq C \Vert q_{0} - \frac{1}{2 \pi} \Vert_{H^{1}} e^{- \gamma t} $ for any 
$ \frac{1-K}{2} < \gamma < 1-K $ and $C=C(\gamma)>0$, and 
and  $P$ is continuous from a neighborhood of $\frac{1}{2 \pi}$ in $H^{1}$ into $E_{\frac{1-K}{2}}$ with the property
$ P(q) = \Pi (q - \frac{1}{2 \pi}) + \mathcal{O} \left ( \left\Vert q - \frac{1}{2 \pi} \right\Vert_{H^{1}}^{2} \right), $
where $\Pi$ is the orthogonal projection onto $E_{\frac{1-K}{2}}$ with kernel $E_{+}$. 
\end{prop}

\begin{proof}
The linearized operator 
 $L_{\frac{1}{2 \pi}} v = \frac{1}{2} \partial_{\theta}^2 v - \frac{K}{2 \pi} J^{\prime} \ast v $
 is diagonal in the classical Fourier basis of  $L^{2}(\mathbb{S})$, 
and self-adjoint with compact resolvent.
Classical proof of local stable manifold existence (see \cite{MR610244} chapter 5 section 1 for example)
 can be adapted here, where there are two identical largest eigenvalues instead of one only. 
\end{proof}

When $K=1$ a bifurcation occurs at $\frac{1}{2 \pi}$. While $\frac{1}{2 \pi}$
still attracts every initial condition (cf Prop.~\ref{Prop_Stabilite_K<=1}), 
solutions do not approach $\frac{1}{2 \pi}$ exponentially fast. 
Nevertheless there is an exponentially attractive (in Gevrey norm)
 two-dimensional {\sl central} manifold, along which a phase is defined: 
for any solution $q(t)$ of equation (\ref{eq}), there is a solution on the central manifold
 $\overline{q}(t)$ -- the {\sl phase} -- such that $q(t)$ converges  to $\overline{q}(t)$ exponentially fast.

\medskip

\begin{prop}
\label{Prop_Stabilite_K=1}
Suppose $K=1$. Let  $k\geq2$, $k \in \mathbb{N}$.
There is a map $\Psi \in C^{k}(\mathcal{O} , E_{+})$ -- $E_+$ equipped with the $H^2$ norm -- with $\Psi(0)=0$ and $D \Psi(0) = 0$, and a neighborhood 
$\mathcal{O}$ of the origin in $E_{\frac{K-1}{2}}$ such that the manifold 
\begin{equation} \mathcal{M}_{0} = \{ u_{0} + \Psi(u_{0}) \; : \; u_{0} \in \mathcal{O} \} 
\end{equation} has the following properties: 
\begin{itemize}
\item [(i)] The central manifold $ \mathcal{M}_{0} $ is locally invariant: 
if $q(t)$ is a solution of equation (\ref{eq}) on $[0,T]$ with 
$q(0) \in \mathcal{O} \cap \mathcal{M}_{0}$ and $ q(t) \in \mathcal{O} $ for all $t \in [0,T]$, 
then $q(t) \in \mathcal{M}_{0}$ for all $t \in [0,T]$.   
\item [(ii)] $\mathcal{M}_{0}$ contains all solutions which stay close enough to $\frac{1}{2 \pi}$ for all time $t \in \mathbb{R}$:
if $q(t)$ is a solution of equation (\ref{eq}) for $t \in \mathbb{R}$ such that 
$q(t) \in \mathcal{O}$ for all $t \in \mathbb{R}$, then $q(0) \in \mathcal{M}_{0}$ and 
$q(t) \in \mathcal{M}_{0}$ for all $t \in \mathbb{R}$.
\item [(iii)] $\mathcal{M}_{0}$ is invariant under reflections and rotations:
if $q \in \mathcal{M}_{0}$ then $\theta \mapsto q( - \theta)$ and 
$\theta \mapsto q(\theta + \varphi)$ (for any $\varphi \in [0, 2 \pi]$) are in $\mathcal{M}_{0}$. 
\item [(iv)] $\mathcal{M}_{0}$ is locally attractive: 
for any $q_{0} \in \mathcal{M}(\mathbb{S})$, 
there is $\overline{q}_{0} \in \mathcal{M}_{0}$ and constants $C$, $\delta > 0$ such that 
\begin{equation} 
\left\Vert S_{t}q_{0} -  S_{t}\overline{q}_{0} \right\Vert_{\mathcal G_a} \leq C  e^{ - \delta t} 
\, .
\end{equation}
\end{itemize} 
\end{prop}
\medskip

\begin{proof}
The results (i) to (iv) are consequences of Theorems 2.9, 3.13 and 3.22 
in \cite{pre05801872} and the Gevrey-convergence
 Lemma (Lemma~\ref{Gevrey_Cv_Lemma}). 
These theorems apply in the setting 
$L_{\frac{1}{2\pi}} : \, H^{2} \rightarrow L^{2}$ is linear continuous and the non linear part $R \in C^{k}(H^{2}, H^{1})$ 
where $H^s$ are the classical Sobolev spaces.
\end{proof}
\medskip

When $K>1$, the incoherent equilibrium at $\frac{1}{2 \pi}$ is no longer stable.
 Classical local stable and unstable invariant manifold theorems apply here, 
and due to the specific form of equation (\ref{eq}), we also have more precise results about invariant manifolds,
 particularly about the way solutions leave $\frac{1}{2 \pi}$.
\medskip

\begin{prop}
\label{Prop_VarStable_K>1}
Suppose $K>1$.
The stable manifold at $\frac{1}{2 \pi}$ is
\begin{equation} 
W^s=W^{s}\left(\frac{1}{2\pi}\right) = \frac{1}{2 \pi} + E_{+} = \{ q \in L^{2}: \, q \geq 0, \, \int q(\theta) d\theta = 1, \,
 \int q(\theta) e^{i\theta} d\theta = 0 \}\, . 
 \end{equation} 
\begin{itemize}
\item [(i)] For any $q_{0} \in W^{s}$, we have 
$ \Vert q(t) - \frac{1}{2 \pi} \Vert_{\mathcal G_a} \leq C \Vert q_{0} - \frac{1}{2 \pi} \Vert_{L^{2}} e^{- \frac{K-1}{2} t}$ for all $t \geq 0$. 
\item [(ii)] Set $\mathcal{O}= B_{L^{2}}(\frac{1}{2\pi}, \delta)\cap \{ \int_\bbS q(\theta) d\theta=1\}$.
 If  $\delta < 1 - \frac{1}{K}$ and 
if a solution  of  (\ref{eq}) satisfies 
$q(t) \in \mathcal{O}$  for all times $t \geq 0$, then $q(0) \in W^{s}$. 
\item [(iii)] More precisely, if for  a $\delta\in ]0,1 - 1/K[$  we have $q(t) \in \cO$ for every $t \in [0,T]$
 then, with $q_0=q(0)$, we have
\begin{equation} 
\label{eq:3ldj}
\left\vert \int q_{0}(\theta) e^{i \theta} d \theta\right\vert = dist( q_{0}, W^{s} ) \leq \delta e^{-\frac{1}{2} [ K(1- \delta) -1] T}\, ,
\end{equation} 
or, conversely, 
if $\vert \int q_0(\theta) e^{i \theta} d\theta \vert \geq \epsilon$ for a value of  $\epsilon \in (0, \delta)$, then there is a time 
$t \leq T_{\delta, \epsilon} = -\frac{2 \ln(\epsilon)}{K(1- \delta) -1}$, such that
$q(t) \notin B_{L^{2}}(\frac{1}{2 \pi}, \delta)$. 
\end{itemize}

Moreover
for any $\alpha \in ]0, \min ( 2, (K-1)/2 [$,
there are constants $C$ and a neighborhood $\widetilde {\mathcal{O}}$ of $\frac{1}{2 \pi}$ in $H^{1}$ such that 
there is a two-dimensional local unstable manifold $W^{u}_{loc}$ Lipschitz continuous in $H^{1}$ 
with the properties
\begin{itemize}
\item  [(iv)]$\frac{1}{2 \pi} \in W^{u}$ and $W^{u}$ has a tangent space
at $\frac{1}{2 \pi}$ which is $E_{\frac{1-K}{2}}$. 
\item [(v)] For any $u_{0} \in W^{u}$, equation (\ref{eq}) has a solution for $t \in ]- \infty, 0]$ with $q(0) = \frac{1}{2 \pi} + u_{0}$ 
such that 
\begin{equation} 
\left\Vert q(t) - \frac{1}{2\pi} \right\Vert_{H^{1}} \leq C \left
\Vert q(0) - \frac{1}{2 \pi} \right\Vert_{H^{1}} e^{\alpha t}\, , 
\end{equation} 
for all $t \leq 0$.
\item [(vi)] And for any $q_{0} \in \widetilde{\mathcal{O}}$ with $\int q_0(\theta) d\theta=1$, if there is a solution $q(t)$ of equation (\ref{eq}) 
which is defined on $]- \infty , 0]$ and satisfies 
$ q(t) \stackrel{t \to - \infty}{\longrightarrow} \frac{1}{2 \pi}$, then $q_{0} \in W^{u}$. 
\end{itemize}
\end{prop}
The fact that $W^u(\frac{1}{2 \pi})$ is two dimensional will in particular be used in section \ref{section_global2}.
\medskip

\begin{proof}
The elementary, but key point is to realize that the dynamics on $W^{s}(\frac{1}{2 \pi})$ is linear: it is just $ \partial _t u= \frac 12 \partial_\theta^2 q$. 
So (i) is a direct consequence of standard results on the linear heat
equation or of Lemma~\ref{Gevrey_Cv_Lemma}.

To prove that any solution staying in a neighborhood of $\frac{1}{2 \pi}$ is in the set $W^{s}$, that is (ii), 
we consider the Fourier ODE system \eqref{eq_FourierODEs} equivalent to equation (\ref{eq}). 
If we set $z_{n} = \rho_{n} e^{i \theta_{n}} \in \mathbb{C}$, we have 
\begin{equation*}
\begin{array}{rl}
 \rho_{n}^{\prime} &= - \frac{n^{2}}{2} \rho_{n} + \frac{Kn}{2} \rho_{1} \left [ \rho_{n-1}\cos(\theta_{1} + \theta_{n-1} - \theta_{n}) 
- \rho_{n+1} \cos(- \theta_{1} + \theta_{n+1} - \theta_{n})  \right ]\, , \\
\rho_{n} \theta_{n}^{\prime} &=  \frac{Kn}{2} \rho_{1} \left [  \rho_{n-1}\sin(\theta_{1} + \theta_{n-1} - \theta_{n}) 
- \rho_{n+1} \sin(- \theta_{1} + \theta_{n+1} - \theta_{n}) \right ]w\, . 
\end{array}
 \end{equation*}
In particular, $ \rho_{1}^{\prime} 
= - \frac{1}{2} \rho_{1} + \frac{K}{2} \rho_{1} \left [ 1 - \rho_{2} \cos( \theta_{2} - 2 \theta_{1}) \right ]. $
Now, observe that $q \in \cO$ implies $\rho_{n} \leq 1$ for all $n \geq 1$ and 
$\rho_{2}(q)\le \delta$. 
So, if $q(t) \in \mathcal{O} $ for all $t \geq T_{0}$, we have 
$ \rho_{1}^{\prime} \geq \frac{1}{2} (K(1-\delta)-1) \rho_{1} $ on $[T_{0} , + \infty[$. 
Since $\delta < 1-\frac1K$ and $K>1$ we have $\rho_1(t) \ge \rho_1(T_0) \exp(c t)$, with $c>0$,
so if $\rho_1(T_0)>0$ we have that $\rho_1(t)$ grows arbitrarily large, which is absurd,
so we must have $\rho_1(T_0)=0$, that is $q(T_0) \in W^s$.

For what concerns (iii)
we use once again that $q \in \cO$ implies $\rho_2 \le \delta$.
Thanks to the Hilbert structure, we also see that $dist(q, W^{s}) = \rho_{1} = (x_{1}^{2} + y_{1}^{2})^{1/2}$. 
Now, suppose that $q(t) \in B(1/2\pi, \delta)$ for all $t \in [0 , T]$: like before 
we obtain $\rho_{1}^{\prime} \geq \rho_{1} \frac{1}{2} \left ( K( 1 - \delta)-1  \right )$, 
and $\rho_{1}(t) \geq \rho_1(0) e^{\frac{K(1- \delta) -1}{2}t}$.
Therefore  $\rho_1(0) \le \rho_1(0) e^{-\frac{K(1- \delta) -1}{2}t}$.
For the second statement in (iii) it suffices to choose 
 $\rho_{1}(0) \geq \epsilon$ and to suppose that $q(t) \in B(1/2\pi, \delta)$ for all $t \in [0 , T]$:
 this leads to a contradiction for $T$ sufficiently large. 

Statements (iv) to (vi) are classical results on the existence of an unstable manifold  in a neighborhood of $\frac{1}{2\pi}$
 (see for example \cite{MR1873467}). See section 5 for further global results on the unstable manifold $W^u(\frac{1}{2\pi})$.
\end{proof}

\subsection{The non-trivial equilibria}

We assume $K>1$.
Several results of \cite{MR2594897} will be used in this section, which we briefly sum up here for readers' convenience. 
Recall (cf \S~\ref{sec:equilibria}) the notation(s) $\hat{q}(\theta)=\hat{q}_0(\theta) = \frac{1}{2 \pi I_{0}(2Kr)} e^{2 Kr \cos(\theta)}$ 
for one of the non-trivial equilibria of  (\ref{eq}) when $K>1$. 
We consider the linearized operator $L_{\hat{q}}$ at $\hat{q}$ 
  $L_{\hat{q}} v = \frac{1}{2} \partial_{\theta}^2 v - K \partial_{\theta} \left [ (J \ast \hat{q}).v + (J \ast v) \hat{q} \right ]\, , $
and the scalar product 
 $\ll u,v \gg = \int_{\mathbb{S}} \frac{\mathcal{U} \mathcal{V}}{\hat q} d\theta $
defined for $u$ and $v$ such that there are a $\mathcal{U}, \mathcal{V} \in L^2(\mathbb{S}, \frac{1}{\hat q} d \theta)$ with 
$\mathcal{U}^{\prime}= u$ and $\mathcal{V}^\prime = v$, and 
$\int_{\mathbb{S}} \frac{\mathcal{U}}{\hat q} d \theta =\int_{\mathbb{S}} \frac{\mathcal{V}}{\hat q} d \theta = 0 $.
The corresponding Hilbert space is denoted $H^{-1}_{1/\hat q}$.
It is not difficult to see that $H^{-1}_{1/\hat q}$ is equivalent to $H^{-1}$, that is the space with weight one 
(see details on this issue in \cite[Section 2]{GPPP}).

\medskip
\begin{thm}(cf \cite{MR2594897})  
\label{th:L}
The operator $L_{\hat{q}}$ is  essentially self-adjoint in $H^{-1}_{1/\hat q}$.
Its spectrum is pure point and lies in $]- \infty, \lambda_1 ] \cup \{ 0 \}$, 
where $\lambda_1<  0$ and $0$ is a simple eigenvalue of $L_{\hat{\hat q}}$ with eigenvector $\partial_{\theta} \hat{q}$. 
Furthermore, for $u,v \, \in D(L_{\hat{q}})$, we have 
\begin{equation} \ll L_{\hat{q}} u, v \gg = \ll u , L_{\hat{q}} v \gg 
= - \frac{1}{2} \int_{\mathbb{S}} \frac{uv}{\hat q} d \theta + \int_{\mathbb{S}} v. (J\ast u) d\theta \, , 
\end{equation}
and  for $K$ fixed we have $D(L_{\hat{q}}^{1/2}) \approx L^{2}_{1/\hat q} (\mathbb{S}) \approx L^{2}(\mathbb{S})$, 
meaning that the scalar products $ \ll L_{\hat{q}}u,v \gg$, $<u,v> = \int_{\mathbb{S}} \frac{uv}{\hat q} d\theta $ and 
$(u,v) = \int_{\mathbb{S}} uv \, d \theta $ 
are equivalent on $R(L_{\hat{q}}) = \{ v \in L^{2}(\mathbb{S}), \, \int_{\mathbb{S}} v d \theta =0 \, \ll v , \partial_{\theta}q \gg =0 \}$.
\end{thm}
\medskip

Of course $\lambda_1=\lambda_1(K)$ depends on $K$. 
In the following we will content ourselves with the fact that $\lambda_1(K)<0 $ for every $K>1$,
 see \cite{MR2594897} for a sharper explicit bound on $\lambda_1(K)$.

\par The next theorem shows that 
the only long term effect of a small $L^{2}$-perturbation of any synchronized equilibria $\hat q_{\psi} \in \mathcal C$ 
is a small rotational shift.
\medskip

\begin{thm} 
\label{Thm_NonLinear_Stability}
There is $\delta>0$ such that if
 $q_0 \in L^2$ with $q_0 \geq 0$ and  $\int_{\mathbb{S}} q_0(\theta) d\theta = 1 $ and if there exists $\psi \in \bbS$ 
such that $\Vert q_{0} - \hat q_\psi \Vert_{L^{2}} \leq \delta $ then
then, for $q(t)$ the solution of equation (\ref{eq}) with $q(0)=q_0$, 
we have $\Vert q(t) -  \hat q_\psi \Vert_{L^{2}} \leq \delta $ for all $t \geq 0$. Moreover
there exists a $\varphi_{\infty} \in \bbS$ such that for any $0<\beta <  \vert \lambda_1\vert$ 
\begin{equation} 
\label{eq:th4.1}
\Vert q(t) - \hat q_{\psi + \varphi_{\infty}} \Vert_{\mathcal G_a} = \mathcal{O} ( e^{- \beta t} )\, , 
\end{equation}
and we have 
\begin{equation} 
\label{eq:th4.2}
\vert \varphi_{\infty} - \varphi_{0} \vert = o \left( \Vert q_0 - \hat q_\psi \Vert_{L^{2}} \right)
 \, \; \textnormal{as } \, \Vert q_0 - \hat q_\psi \Vert_{L^{2}} \rightarrow 0\, . 
 \end{equation}
\end{thm}

From the quoted theorem of \cite{MR2594897} of from \ref{Thm_NonLinear_Stability}, 
we can deduce that $\mathcal{C}$ is an invariant and stable normally hyperbolic curve for (\ref{eq}). 
 (see definition in \cite{MR1602566} or \cite{MR1873467} for instance). 

\begin{corollaire}\label{Corollaire_Normal_Hyperbolicity} 
The smooth and invariant curve $\mathcal{C}$ is stable normally hyperbolic.
More precisely, at each point $q_{\varphi} \in \mathcal{C}$, 
we consider the splitting 
\begin{equation} 
L^{2}= T_{\hat q_{\varphi}}\mathcal{C} \oplus N_{\varphi}^{s},
\end{equation}
with $T_{\hat q_{\varphi}}\mathcal{C} = \mathbb{R}. \partial_{\theta} \hat q_{\varphi} $ is the kernel of the linearized 
operator $L_{\hat q_{\varphi}}$ at $\hat q_{\varphi}$, and $N^{s}_{\varphi}$ is its orthogonal with respect to 
the $H^{-1}_{1/\hat q_{\varphi}}$ scalar product 
($N^{s}_{\varphi}$ is also the range of the operator $L_{\hat q_{\varphi}}$).
This splitting depends continuously on $\varphi$ 
and there is a $T > 0$ such that
\begin{equation} 
\Vert D S^{nT}(\hat q_{\varphi}).v \Vert \leq \Vert v \Vert  
\;\; \text{ for any}\;\; v \in T_{\hat q_{\varphi}} \mathcal{C}  \, \text{ and}\;\; n \in \mathbb{Z}, \, \text{ and} \end{equation}
\begin{equation} \Vert D S^{nT}(\hat q_{\varphi}) . v \Vert \leq \frac{2}{10^{n}} \Vert v \Vert 
\;\; \text{for any} \;\; v \in N^{s}_{\varphi}, \, \text{ and}\;\; n \geq 1.\end{equation}
\end{corollaire}

The proof of the theorem is outlined here, details and the proof of the corollary are postponed to the appendix. 
We mention that using that the invariant curve $\mathcal C$ is in fact not only invariant but a composed of equilibria only, 
its normal hyperbolicity can also be derived from the spectral gap result of \cite{MR2594897} directly, see \cite{GPPP}.

\begin{proof}
Thanks to rotation-symmetry of equation (\ref{eq}), without loss of generality
we can assume that $\psi=0$, ie $\hat q_\psi= \hat{q}$ and $\Vert q_0 - \hat{q} \Vert_{L^2} < \delta$.

\begin{lm}\label{lm:new_local_coordinates}
Consider $u(t, \theta) = q(t, \theta) - \hat{q}(\theta)$ and 
$x_{\varphi} (\theta) = \hat q_{\varphi}(\theta) - \hat{q}(\theta) = \hat{q}(\theta + \varphi) - \hat{q}(\theta) $.
There is a  neighborhood $\mathcal V$ of $0$ in $L^2$ and to smooth projections 
$$\begin{array}{rcl}
\mathcal V & \rightarrow & ]- \epsilon, \epsilon[ \subset \mathbb R \\
  u & \mapsto & \varphi(u)
\end{array} \; \text{ and } \;
\begin{array}{rcl}
\mathcal V & \rightarrow & R(L_{\hat q}) \subset L^2 \\
  u & \mapsto & y_u, 
\end{array}$$
such that $u = x_\varphi(u) + y_u $ for all $u \in \mathcal V$, 
and $u + \hat q$ is a solution of equation (\ref{eq}) if and only if $\varphi = \varphi(u)$ and $y=y_u$ are solutions of 
\begin{equation}  \frac{d \varphi}{dt} = \Phi(\varphi , y)  \, \; \textnormal{ and } \, \;
 \frac{dy}{dt} + A_{2} y = g(\varphi , y)\, , \end{equation}
where 
\begin{equation} 
\Phi(\varphi , y) = \frac{1}{\ll v , \partial_{\varphi} x_{\varphi}\gg} \ll v , f(x_{\varphi} + y) - f(x_{\varphi})\gg  \, ,
\end{equation}
 \begin{equation} 
 g(\varphi, y ) =f(x_{\varphi} + y) - f(x_{\varphi}) - \partial_{\varphi}x_{\varphi} \Phi(\varphi, y) ,
\end{equation} 
are smooth, with $D \Phi(0,0)=0$ and $Dg(0,0)=0$, 
and $A_{2}$ is the restriction of $L_{\hat{q}}$ on $R(L_{\hat{q}})$, ie  $A_{2} = L_{\hat{q}} \vert_{R(L_{\hat{q}})}$.
\end{lm}
\begin{proof}
Details of proof of this lemma can be found in the appendix. 
\end{proof}

From this we can deduce a local stability result for the $y$ variable.
Let us introduce the natural Hilbert basis of eigenvectors $e_{k} \in D(L_{\hat{q}})$ so that 
$L_{\hat{q}}= \sum_{k \geq 0} \lambda_{k} \ll \cdot , e_{k} \gg e_{k}\, ,$
with $\lambda_{0} = 0$ and $0 > \lambda_1 \ge \lambda_{n} \ge  \lambda_{n+1} \rightarrow - \infty $ for $n \geq 1$. 
We have $A_{2}= \sum_{k \geq 1} \lambda_{k} \ll \cdot , e_{k} \gg e_{k}$ and 
$e^{tA_{2}} = \sum_{k \geq 1} e^{t \lambda_{k}} \ll \cdot , e_{k} \gg e_{k}$, so that
\begin{equation}
\Vert A_{2}^{\alpha} e^{tA_{2}} v \Vert^{2}_{H^{-1}} \leq 
e^{(1-\epsilon) \lambda_1 t} \max_{k} \left ( \vert \lambda_{k} \vert^{2\alpha} e^{-2 \epsilon \vert \lambda_{k} \vert t} \right ) 
\Vert v \Vert^{2}_{H^{-1}}\, , 
\end{equation}
for any $\epsilon \in ]0, 1[$
and from this we directly infer
\begin{equation}\label{eq:estim_A_2}  
\Vert e^{tA_{2}}v \Vert_{D(A_{2}^{\alpha})} \leq 
\frac{C_{\alpha,\epsilon}}{t^{\alpha}} e^{- (1-\epsilon)\vert \lambda _1 \vert  t} \Vert v \Vert_{H^{-1}} \end{equation}
for any $t>0$ and $\alpha >0$ (we have used 
$\Vert v \Vert_{D(A_{2}^{\alpha})} ^2$ for $\Vert A_2^\alpha v \Vert ^2_{H^{-1}}$). 

\begin{lm}
\label{th:lem4}
For every $ \beta \in ]0, \vert \lambda_1 \vert [$
there exist $\rho_{0} > 0$ and $M>0$ such that, 
if $\vert \varphi_{0} \vert + \Vert y_{0} \Vert_{L^{2}}  = \frac{1}{8} \rho \leq \frac{1}{8} \rho_{0} $, 
then  for all $t \geq 0$
\begin{equation}\label{th:lem4.1}
\vert \varphi(t) \vert + \Vert y(t) \Vert_{L^{2}} \leq \rho\, .
\end{equation}
Moreover for all $t \geq 0$ we have
 \begin{equation} 
 \label{eq:lem4.2}
 \Vert y(t) \Vert_{L^{2}} \leq 2 \Vert y_{0} \Vert_{L^{2}} e^{- \beta t}\, . 
 \end{equation}
\end{lm}
\begin{proof} 
Details of the proof of this lemma can be found in the appendix. 
\end{proof}

Since $\Phi$ is smooth, we have
$ \left\vert \frac{d \varphi}{ dt}(t) \right\vert \leq C \Vert y_{0} \Vert e^{- \beta t}, $
 and
$\varphi(t)$ converges: $\vert \varphi(t) - \varphi_{\infty} \vert = \mathcal{O} \left ( e^{- \beta t} \right )$. 
So we have 
\begin{equation} \Vert y(t) \Vert_{L^{2}} + \vert \varphi(t) - \varphi_{\infty} \vert = \mathcal{O} \left ( e^{- \beta t} \right ), \end{equation}
which, combined with the Gevrey-convergence lemma, gives \eqref{eq:th4.1}. 

For \eqref{eq:th4.2} we argue instead that
 $ \vert \frac{d \varphi}{dt} \vert \leq {2\gamma(\rho)} \Vert y_{0} \Vert_{L^{2}} e^{- \beta t} $
and that $\vert \varphi_{\infty} \vert $ is bounded by $ {2\gamma(\rho)} \Vert y_{0} \Vert_{L^{2}} \frac{1}{\beta} 
= o \left ( \vert \varphi_0 \vert + \Vert y_0 \Vert \right )$, for some $\gamma(\rho)$ satisfying $\gamma(\rho) \underset{\rho \rightarrow 0}{\longrightarrow} 0$
(see the proof of lemma \ref{th:lem4} in the appendix), 
and $o \left ( \vert \varphi_0 \vert + \Vert y_0 \Vert \right ) = o \left ( \Vert q_0 - \hat{q} \Vert \right ) $ 
since $ u \mapsto (\varphi, y) $ is Lipschitz and its inverse is Lipschitz too.
This completes the proof of Theorem ~\ref{Thm_NonLinear_Stability}.
\end{proof}

\section{The global attractor}
\label{section_global}

This section is entirely devoted to the study of the global attractor of equation (\ref{eq}). 
It starts with general results and progressively moves to more refined descriptions of this set. 
The end of the section proposes conjectures supported by numerical investigations.

\subsection{Existence}
\label{section_global1}

The existence of a global attractor is a classical consequence of the existence of a compact absorbing set, 
and regularizing properties of equation (\ref{eq}) imply its existence (see \cite{MR2594897,MR2164412,MR1691574} for example). We state it here :
\begin{thm} \label{Thm_Existence_Attracteur}
There is a set $\mathcal{A}$ bounded in Gevrey space $\mathcal G_a$, such that for any space $E$ such 
that the injection $\mathcal G_a \rightarrow E$ is compact, for any bounded $B$ set in $E$, we have 
\begin{equation} dist (S_t B, \mathcal{A} ) \underset{t \rightarrow + \infty}{\longrightarrow} 0 
\end{equation}
where $dist$ is measured in the natural $E$ norm. 
This includes the cases $E= L^2$, $E= H^s$ for any $s >0$, $E=C^k$ for any $k \geq 1$,  and even $E=G_{a'}$ if $ a' < a$.
\end{thm}

\par When $K\leq 1$, the global attractor is $\{ \frac{1}{2\pi} \}$. In the remainder of section \ref{section_global}, 
we discuss the global attractor in the case $K>1$. 

\subsection{Description of the global attractor}
\label{section_global2}

\par Equation (\ref{eq}) is equivariant under $\mathcal O (2)$, and the subspace of even functions is invariant by the dynamics. 
However for $q \in L^{2}(\mathbb{S})$ there is in general no $\alpha \in [0 , 2\pi[$ such that $q( \cdot + \alpha)$ is even. 
More precisely, for $q = \frac{1}{2\pi} + \frac{1}{2\pi} \sum x_k \cos(k\theta) + y_k \sin(k\theta)$, such an $\alpha$ exists 
if and only if the quantities $\frac{x_k}{\sqrt{x_k^2 + y_k^2}}$ and $\frac{y_k}{\sqrt{x_k^2 + y_k^2}}$ are independent of $k$.
So the whole dynamics of (\ref{eq}) in $L^{2}$ are not captured by its restriction on the subspace of even solutions. 
Nevertheless, we show below that the dynamics on the global attractor $\mathcal{A}$ are characterized by its restriction on the subspace of even functions. 
In fact the global attractor is radial:  it is composed of one even heteroclinic solution $q_h$, with $q_{h}(- \infty) = \frac{1}{2\pi}$ 
and $q_{h}(+ \infty) = \hat{q}$, and its rotations $q_{h}(t, \cdot + \varphi)$. 

\begin{prop}\label{prop:GA_disk}
\par The unstable manifold $W^{u}(\frac{1}{2 \pi})$ is the global attractor of the system. 
\par The 2 dimensional unstable manifold of the constant equilibrium $\frac{1}{2 \pi}$ consists in a family of heteroclinic orbits, 
each one connecting $\frac{1}{2 \pi}$ to a non trivial equilibrium $\hat{q}( \cdot + \varphi)$ ($\varphi \in [0, 2 \pi[$). 
All these orbits are obtained by rotation from one heteroclinic connection in the even functions subspace. 
In particular, for any solution $q$ in $W^{u}(\frac{1}{2 \pi})$ of equation (\ref{eq}), 
there is an angle $\varphi$ independent of the time $t$, such that 
$q(t, \cdot + \varphi)$ is even for any $t \in ]- \infty , + \infty[$.
\end{prop}
\begin{proof}
\par In dissipative systems unstable manifolds of equilibria are included in the global attractor \cite{MR941371}. 
The converse inclusion is a rather general property, it relies essentially on the existence of a Lyapunov functional 
and on the fact that $\frac{1}{2\pi}$ is the {\em only} equilibrium with unstable directions. 
It implies that the global attractor $\mathcal{A}$ is exactly the unstable manifold $W^u(\frac{1}{2\pi})$ here.
\par We show now that in the even subspace, the 1-dimensional unstable manifold of the constant equilibrium $\frac{1}{2 \pi}$ consists of two heteroclinic orbits, 
connecting $\frac{1}{2 \pi}$ to $\hat{q}$ and $\hat{q}( \cdot + \pi) $.
Up to the symmetry $\theta \mapsto - \theta$ and time translations, with no loss of generality we can consider only 
one solution $q(t)$ escaping from $\frac{1}{2 \pi}$ with $x_{1}(t) >0$ when $t \rightarrow -\infty$. 
The dissipativity and regularity properties imply that the trajectory $\{ q(t), \; t \in \mathbb{R} \}$ is compact and
 that $q(t)$ it converges to one of the 3 equilibria of the system $\hat q$, $\frac{1}{2\pi}$ or $\hat q(\cdot + \pi)$. 
One has the estimate $\vert x_{k}(t) \vert \leq 1$ for all Fourier coefficients along solutions of our equation, and
we can easily see that $x_{1}(t) > 0$ for all time $t \in \mathbb{R}$. 
From the Fourier ODEs system we see that $\underset{t \rightarrow + \infty}{\lim} x_{1}(t) = 0$ is absurd if $x_1$ is not $0$ for all times.
Hence the solution $q(t)$ is not an homoclinic connection of $\frac{1}{2 \pi}$, and it is a heteroclinic solution.   
The $\mathcal O (2)$ invariance of equation (\ref{eq}) eventually implies the second part of proposition \ref{prop:GA_disk}. 
\end{proof}

\subsection{Global attractor and numerical explorations}
\label{section_global3}

\par As in the previous section, we denote $q_h$ the heteroclinic solution of (\ref{eq}) such that 
$q_h(t) \underset{t \rightarrow - \infty}{\longrightarrow} \frac{1}{2\pi}$ and $q_h(t) \underset{t \rightarrow + \infty}{\longrightarrow} \hat q$, 
where $\hat q (\theta)$ is the unique equilibrium of (\ref{eq}) with $x_1 >0$ in the even functions subspace. 
Since the structure of the global attractor is radial, without loss of generality, we can restrict ourselves to 
working in the subspace of even functions. 

This section is organized as a sketch of a proof of two conjectures, which are partially proved and partially supported by numerical evidences. 
The first one is that the global attractor is not only a two dimensional manifold homeomorphic to a disk, but it is an analytical graph 
over the disk  $\{ x_1^2 + y_1^2 \leq 1 \}$ in $L^2$, and an recursive procedure for the Taylor series coefficient is given. 
The second conjecture is that this graph, which depends on the coupling constant $K$, 
converges in analytic functions space when $K \rightarrow + \infty$ to Ott-Antonsen two-dimensional invariant manifold. 

\subsubsection{The global attractor is an analytic graph}

Consider the global attractor as a geometrical curve $ \mathcal{H} = \{ q_h(t), t \in ]- \infty , + \infty[ \}$. 
\begin{prop}\label{prop:graph}
There is a neighborhood $\mathcal V$ of $\frac{1}{2 \pi}$ in $L^2$ such that 
$\mathcal H \cap \mathcal V $ is a graph on $\vspan \{  \cos  \}$. 
That is to say there is $\epsilon >0$ such that
\begin{equation}\label{eq:bij_x1}
 \begin{array}{rcl}
\mathcal H \cap \mathcal V & \rightarrow & ]0, \epsilon[ \\
q & \mapsto & x_1 = \int_{- \pi}^\pi q(\theta) \cos(\theta) d\theta 
\end{array} 
\end{equation}
 is a smooth bijection. 
\par Moreover, if $t \mapsto x_1(t) = \int_{- \pi}^\pi q_h(t,\theta) \cos(\theta) d\theta $ is increasing on $]- \infty, + \infty[$ 
($q_h$ denotes the heteroclinic solution of (\ref{eq})), 
then $\mathcal H$ is globally a $C^\infty$ graph on $\vspan \{  \cos  \}$.
\end{prop}
\begin{proof}
By unstable manifold theorem at $\frac{1}{2\pi}$, we know that $\mathcal{H}$ is a $C^\infty$-graph over $\mathbb{R}.\cos$ 
in a neighborhood of $\frac{1}{2\pi}$, and
there is a $T_0 > - \infty$ such that $t \mapsto x_1(t)$ is increasing on $]- \infty, T_0[$. 
The first part of the proposition holds in particular for 
$\mathcal V = \{ q \in L^2 , \; 0< \int_{- \pi}^\pi q(\theta) \cos(\theta) d\theta < x_1(T_0)   \}$.
If we have $T_0 = + \infty$, then $t \mapsto x_1(t)$ is a bijection from $\mathbb R$ to its image, 
and $\mathcal H$ is globally a graph on $\vspan \{  \cos  \}$.
\end{proof}
In the following we denote the graph function by 
$ \Phi_\mathcal H : \left \{ \begin{array}{rcl}
]0,\epsilon[ & \rightarrow & \mathcal H  \\
x_1   & \mapsto & q
\end{array} \right .$  inverse of the bijection (\ref{eq:bij_x1}).

\begin{num_res} 
For all $t \in \mathbb R$, $x_1^\prime(t) > 0$. (see figure \ref{Figx1})
\end{num_res}
The property $ x^\prime_n(t) > 0$ is numerically true for any $n \geq 1$ and any time $t \in \mathbb R$ along the solution $q_h$.
All these curves have a unique inflexion point, they are convex for large negative times $t$ and concave when $t \rightarrow + \infty $. 
Their limit, the Fourier coefficients of $\hat{q}$, are given by the Bessel modified functions :
 $ \underset{t \rightarrow + \infty}{\lim} x_n(t) = \frac{I_n(2Kr)}{I_0(2Kr)}, $
so that all of them tend to $ \underset{x \rightarrow + \infty}{\lim} \frac{I_n(x)}{I_0(x)} = 1$ when $K \rightarrow + \infty$, 
but for any finite $K$ we have $  \underset{t \rightarrow + \infty}{\lim} \vert x_n(t) \vert \leq C \frac{1}{a^n} $ for some constants $C >0$, 
$a >1$ and all $n$ large enough. \\

\begin{figure}[h!tp]\label{Figx1}
\begin{center}
\psfrag{e}[B][B][1][0]{ \large {$t$}}
\psfrag{(a)}[B][B][1][0]{ \large {(a)}}
\psfrag{e log(T)}[B][B][1][0]{ \large {$x_1(t)$}}
\psfrag{K=2.0}[B][B][1][0]{  \small{$K=2.0$}}
\psfrag{K=4.0}[B][B][1][0]{  \small{$K=4.0$}}
\includegraphics[scale=0.85]{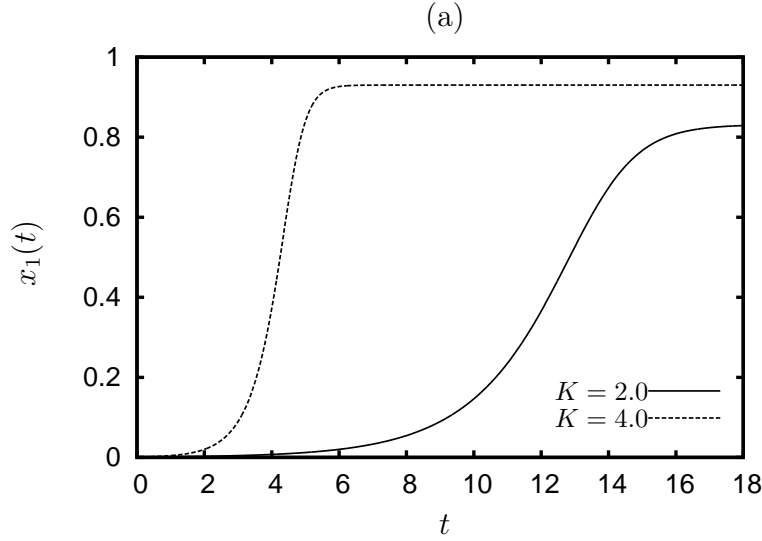}
\end{center}
\begin{center}
\vskip -4mm
\caption{The heteroclinic solutions $q_h$ is a graph on $x_1$. 
\small{Plots of $x_1(t)$ along the heteroclinic solution $q_h$ for $K=2.0$ (bottom) and $K=4.0$ (top).
These curves were obtained using the Fourier ODEs system equivalent to (\ref{eq}), truncated with $N=200$ equations, 
Euler scheme and $dt=0.00001$. Numerics were checked using $N=100$ and $dt=0.0001$ and classical Runge Kutta forth order scheme, 
the maximum difference between the two is of order $10^{-5}$.
The first Fourier coefficient $x_1$ is increasing along the heteroclinic solution. The same holds for all the $x_n(t)$, $n \geq 1$ (Figures not shown).}}
\end{center}
\end{figure}

\par Recall that $\int_{- \pi}^\pi \hat q(\theta) \cos(\theta) d\theta =  r $ (see section (\ref{sec:equilibria})), so that
 $ \underset{t \rightarrow + \infty}{\lim} \int_{- \pi}^\pi  q_h(\theta) \cos(t,\theta) d\theta = \underset{t \rightarrow + \infty}{\lim}  x_1(t) = r$.
If proved, the monotonicity of $x_1$ would imply the following. 
\begin{conj}
The graph function $\Phi_\mathcal H$ is $C^\infty([0,r])$, 
and the global attractor $W^\frac{1}{2\pi}$ is a $C^\infty$ graph on the disk $\{ (x_1, y_1), \; x_1^2 + y_1^2 < r \}$.
\end{conj}
\medskip

\par With the notations of proposition \ref{prop:graph} above, 
there are functions $\varphi_n \in C^{\infty}$ such that for all $x_1 \in ]0, \epsilon[$, we have 
 $ x_n(t)= \int_{- \pi}^\pi q_h(t,\theta) \cos(n \theta) d\theta = \varphi_n(x_1(t)).$
So for $t \in ]- \infty, T_0[$ equation (\ref{eq}) is equivalent to 
\begin{equation} x_1^{\prime} = - \frac12 x_1 + \frac{K}{2} x_1 (1- \varphi_2 (x_1)) \; \; \textnormal{and } \, \; 
x_n = \varphi_n (x_1), \, \; \forall n \geq 2.
\end{equation}

\begin{prop}
The functions $\varphi_n$ are $C^{\infty}$ in a neighborhood of $0$, 
their Taylor series are of the form $$ \varphi_{n}(x) = x^{n} \sum_{p\geq 0} \alpha_{n,2p} x^{2p} ,$$ 
where the coefficients $\alpha_{n,2p}$ can be computed recursively.
\end{prop}
\begin{proof}
We use the notation $\varphi_1(x) = x$, so that $\alpha_{1,0} = 1$ and $\alpha_{1,p}= 0$ for all $p \geq 1$.
We use the Fourier ODE system 
$x_n^{\prime} = -\frac{n^2}{2} x_n + \frac{Kn}{2} x_1 ( x_{n-1} - x_{n+1}) \;  (\forall n \geq 1) $,
equivalent to equation (\ref{eq}), and the Taylor expansion of the $\varphi_n$ to obtain relations between the $\alpha_{n,p}$.
First one finds that $p<n$ implies $\varphi_n^{(p)} = 0$, so that $\varphi_n(x) = x^{n} \sum a_{n,p} x^{2p}$, and then all odd 
coefficients $a_{n, 2p+1}$ are zero, so that the series has the form $ \varphi_{n}(x) = x^{n} \sum_{p\geq 0} \alpha_{n,2p} x^{2p}$.

Then one finds 
$\alpha_{2,0}= \frac{K}{K+1} \; \textnormal{and } \;
 \alpha_{n,0} = \frac{K}{K + (n-1)} \alpha_{n-1,0} = \Pi_{j=2}^{n-1} \frac{K}{K + j} $
and the recurrence relations 
\begin{equation}\label{eq_rec1}  ((1+p) K +(1-p)) \alpha_{2,2p} = 
\frac{K}{2} \left ( \sum_{i+j=p-1, i\geq0, j\geq0} (2+2j) \alpha_{2,2j} \alpha_{2,2i} \right )
- K \alpha_{3, 2(p-1)}  \end{equation}
for $n=2$ and all $p \geq 1 $ and
\begin{equation}\label{eq_rec2}
\begin{split} 
 \left( (n+2p)K + (n^{2} -n - 2p) \right) \alpha_{n,2p} &=  \\
K\left( \sum_{i+j=p-1} (n+2j)  \alpha_{n,2j} \alpha_{2,2i} \right)&
 + Kn ( \alpha_{n-1,2p} - \alpha_{n+1, 2(p-1)} ) 
 \end{split}
\end{equation}

for all $n \geq 3$ and $p \geq 1 $.
The no-resonance conditions 
$ (n+ 2p) +  \frac{n^{2} - n -2p}{K} \neq 0 $ are satisfied as soon as $K > 1$.

We introduce the order $ \ll $ on $\mathbb{N}^2$ : 
we have $(m_1, q_1) \ll (m_2, q_2)$ iff ($q_1 < q_2$ or ($q_1 = q_2$ and $m_1 < m_2$)). This is the lexicographical order for 
the inverted couples $(q_i, m_i)$. One can check that for that order, the recursion relation defining $\alpha_{n,2p}$ 
uses only $\alpha_{m,2q}$ with $(m,2q) \ll (n ,2p)$. 
There are no infinite strictly decreasing sequences in $\mathbb{N}^{2}$ for the order $\ll$, 
and so one can solve recursively the relations defining the $\alpha_{n,2p}$. 
\end{proof}

\par Consider $q_h$ an heteroclinic solution of (\ref{eq}) and $x_n(t)$ its Fourier coefficients. 
Since $t \mapsto x_1(t)$ is increasing and bijective $\mathbb R \rightarrow ]0,r[$, we can define 
$X_n(x)$ for all $x \in ]0,r[$ by $X_n( x_1(t) ) = x_n\left( x_1^{-1}( x_1(t) ) \right) $ for all $t \in \mathbb R$.

\begin{num_ress}
For any $K$ and any $n$, the Taylor series defining $\varphi_n$ has a positive radius, which is $r$ 
(see figure \ref{Fig2}).
For any $K$ and any $n$, the sum of the Taylor series defining $\varphi_n$ equals $X_n$ on $]0,r[$ (see figure \ref{Fig4}).
\end{num_ress}
For a series $\sum_1^\infty a_p x^p$, the convergence radius $R$ is given by 
$R^{-1} = \underset{p \rightarrow + \infty}{\limsup} \vert a_p \vert^{\frac 1p}$.
In figure \ref{Fig2}, we see that 
$r_{n,p} = \log ( \left( \alpha_{n,2p} \right)^{\frac{1}{p+1}}) \underset{p \rightarrow + \infty}{\longrightarrow}~l~>~0$, 
where $l> 0$ is a real number depending on $K$ but not on $n$. That is to say all series $\sum_p \alpha_{n,2p}x^{2p}$ have a positive radius
$R = e^{- l/2}$,
which depends on $K$ but does not depend on $n$. We also observe that for all $K$, the limit is $l \approx -2 \log(r)$, 
suggesting that the radius $R= r$. Since the Taylor series of $\varphi_n$ are $x^n \sum_p \alpha_{n,2p}x^{2p}$ the same results hold for them. 
We have also checked that d'Alembert's ratio test is numerically consistent with this: 
for all $n \geq 2$ and $K > 1$ we have 
$\frac{\alpha_{n,2p+2}}{\alpha_{n,2p}} \underset{p \rightarrow + \infty}{\longrightarrow} \tilde l$, 
and $\tilde l \approx \frac{1}{r^2}$  suggesting that $R=r$ too (figures not shown).
Figure  \ref{Fig4} shows that the troncated sum of the Taylor series $\varphi_n \approx \sum_0^N \alpha_{n,2p} x^{n + 2p}$ 
equals $X_n$ on $]0,r[$. Notice that the verification $\sum_{p} \alpha_{n,2p} x^{n + 2p} = X_n(x)$ is necessary since some functions, 
like $x \mapsto e^{- \frac{1}{x^2}} $, have converging Taylor series without being analytical. 

\begin{figure}[hhh!tp]
\begin{center}
\psfrag{(a)}[B][B][1][0]{ \small {(a)}}
\psfrag{(b)}[B][B][1][0]{ \small {(b)}}
\psfrag{p}[B][B][1][0]{ \tiny {p}}
\psfrag{rp}[B][B][1][0]{ \small {$r_{n,p}$}}
\includegraphics[scale=0.55]{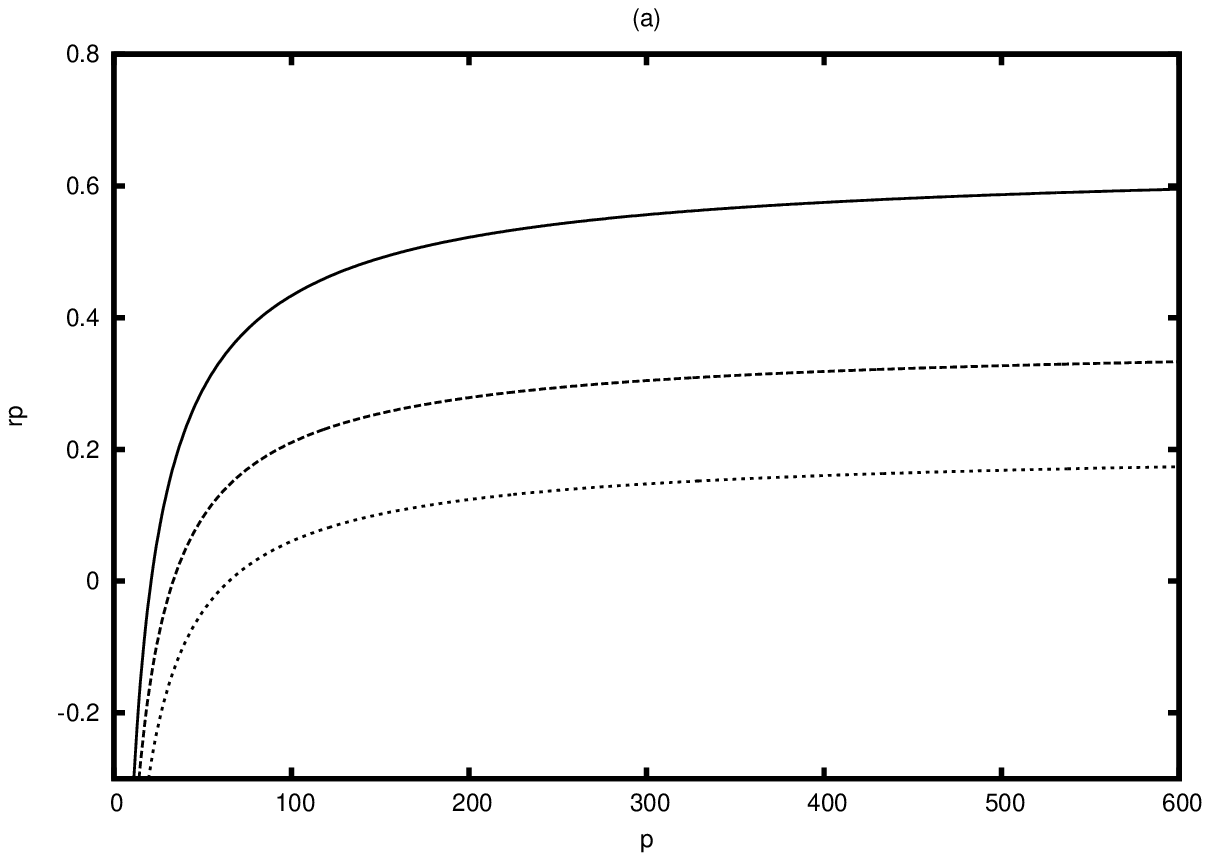}
\psfrag{(c)}[B][B][1][0]{ \small {(b)}}
\psfrag{p}[B][B][1][0]{ \tiny {p}}
\psfrag{rp}[B][B][1][0]{ \small {$r_{n,p}$}}
\includegraphics[scale=0.55]{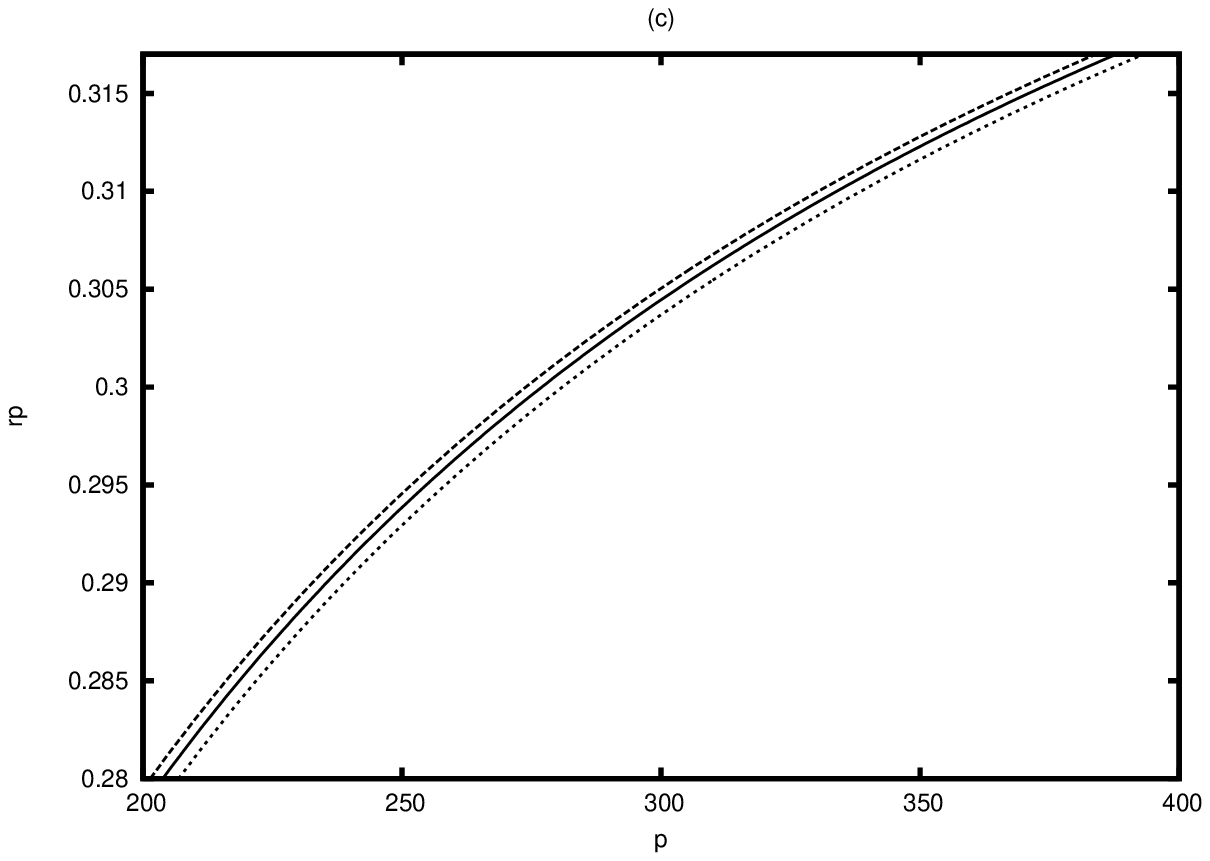}
\end{center}
\vskip -2mm
\caption{Positive convergence radius of the series $\sum_{p \geq 0}\alpha_{n,2p}x^{2p + n}$.
\small{Plots of $r_{n,p}= \frac{1}{p+1} \log(\alpha_{n,2p})$ for different $n$ and different values of $K$ :
(a) : for $n=2$, and from top to bottom $K=1.5$, $K=2.0$ $K=3.0$,
(b) : for $K=2.0$, and from top to bottom $a=3$, $a=2$ $a=4$. 
Notice that (a) and (b) show $r_{n,p}$ in different windows.
The radii of the series are positive since the sequences $r_{n,p}$ converge to positive values, 
with $\lim_{p \rightarrow + \infty} r_{n,p} = -2 \log (r)$. 
For any value of $K$, the radius $R(\varphi_n)$ does not seem to depend on $n$.
The radius increases when $K$ increases, but the convergence 
$r_{n,p} \stackrel{p \rightarrow  \infty}{\longrightarrow} -2 \log (r)$ is slower for large $K$. }
}
\label{Fig2}
\end{figure}

\begin{figure}[hhh!tp]
\begin{center}
\psfrag{(a)}[B][B][1][0]{ {(a)}}
\psfrag{x1}[B][B][1][0]{ {$x_1$}}
\psfrag{x2=phi2}[B][B][1][0]{ {$x_2=\varphi_2(x_1)$}}
\psfrag{x2(x1)}[B][B][1][0]{ \tiny{$x_2$}}
\psfrag{phi2}[B][B][1][0]{  \tiny{$\varphi_2$} }
\includegraphics[scale=1.2]{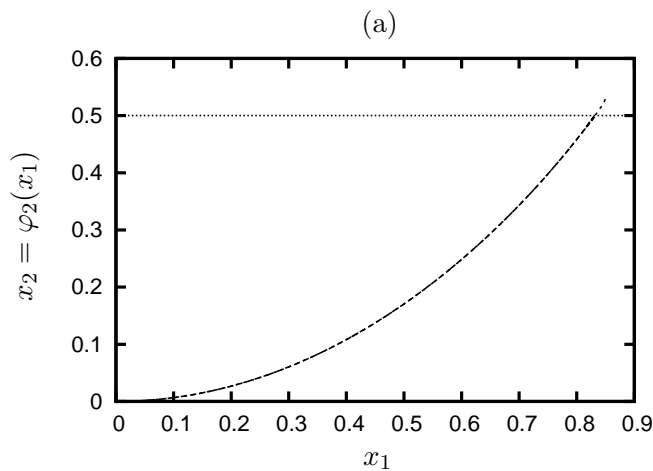} 
\end{center}
\vskip -3mm
\caption{The sum $\sum_{p}\alpha_{n,2p}x^{n+2p}$ coincides with the function $x_n~=~x_n(x_1)$.
\small{Plots of functions $y=\varphi_2(x)$ and $y(x)=x_2(x_1^{-1}(x))$. The sum of the series was computed with $N=200$ terms, 
and the drawn curve $ \{(x,y(x)), \; x~\in~[0,r]\}=\{(x_1(t), x_2(t)), \; t \in \mathbb{R} \}$ was computed with the Fourier ODEs system with $N=200$ equations, Euler scheme and $dt=0.00001$.
The numerical difference between the two is less than $10^{-5}$ for $x_1 \in [0,r]$.
The equilibria $\hat{q}$ is at $x=r$, characterized by $\varphi_2(r)= 1 -\frac{1}{K}= \frac12$ here.
The sum of the Taylor series coincides with the function $x_n = x_n(x_1)$ along the heteroclinic solution $q_h$ for all $n \geq 2$ .
In particular, this heteroclinic solution seen has a graph on $x_1$ is an analytic function on $[0,r]$.}
}
\label{Fig4}
\end{figure}

{
The previous numerical observations suggest the following.
The graph function satisfies
$$\Phi_\mathcal H(x) = \frac{1}{2\pi} + \frac 1\pi \sum_{n \geq 1} \varphi_n(x) \cos(n \cdot)
= \frac{1}{2\pi} + \frac 1\pi \sum_{n \geq 1} \left ( x^n \sum_{p=1}^{+ \infty} \alpha_{n,2p} x^{2p}  \right ) \cos(n \cdot),$$ 
and it is analytical on $[0,r[$ in the sense that all $\varphi_n$ are.
\begin{conj}
\par The global attractor $\mathcal H$ is an analytical graph on the disk $\{ (x_1,y_1), \; x_1^2 + y_1^2 < r  \}$.
\end{conj}
}
\par These Taylor series can also be used to estimate the eigenvalues at 
the incoherent equilibrium $\frac{1}{2\pi}$ and at the synchronized equilibrium $\hat{q}$.
 At $\frac{1}{2\pi}$ one can analytically compute the eigenvalues, and we have checked that 
the numerical estimates given by the Taylor series are accurate. At $\hat{q}$, the eigenvalues can also be computed 
by linearizing the truncated Fourier ODE system (\ref{eq_FourierODEs}), and we have checked that the two methods agree there too. 
\par The recurrences equations (\ref{eq_rec1}) and (\ref{eq_rec2}) can actually be used to compute recursively the 
Taylor coefficients $\alpha_{n,2p}$, and so one can compute directly the heteroclinic solution 
$$ q_h(x_1) = \frac{1}{2\pi} + \frac{1}{\pi} \sum \varphi_n(x_1) \cos(n\theta) $$ 
and the whole global attractor with arbitrary precision without discretizing the partial differential equation (\ref{eq}) 
or the equivalent infinite Fourier system (\ref{eq_FourierODEs}).

\medskip

\subsubsection{Convergence to Ott-Antonsen Ansatz for large $K$}

\begin{prop}\label{prop:cv_form}
Considering the Taylor series $\varphi_n (x)= x^n \sum_{p \geq 0} \alpha_{n,2p} x^{2p} $,
we have $\alpha_{n,0} \underset{K \rightarrow + \infty}{\longrightarrow} 1$ and 
$\alpha_{n,2p} \underset{K \rightarrow + \infty}{\longrightarrow} 0$ for all $n \geq 2$ and all $p \geq 1$. 
\end{prop}
\begin{proof}
The formulas $\alpha_{n,0} = \Pi_{j=2}^{n-1} \frac{K}{K + j} $ found above show that $\alpha_{n,0} \rightarrow 1$ when $K \rightarrow + \infty$.

When $p=1$ we have 
\begin{equation}  \left ( \left (1+\frac{2}{n} \right ) + \frac{n^{2} -n - 2}{Kn} \right ) \alpha_{n,2} = 
 \alpha_{n,0} \alpha_{2,0}  +  ( \alpha_{n-1,2} - \alpha_{n+1, 0} ) \end{equation}
and knowing $\alpha_{n,0} \rightarrow 1$, the recursion hypothesis $\alpha_{n,2} \rightarrow 0$ for all $n \leq N$ implies $\alpha_{N+1,2} \rightarrow 0$.
One can directly check that $\alpha_{2,2} \rightarrow 0$, and so deduce $\alpha_{n,2} \underset{K \rightarrow + \infty}{\longrightarrow} 0$ for all $n \geq2$. 

When $p \geq 2$ (and $n\geq 2)$, using what we know about the case $p=0$ and $p=1$, and the relation 
\begin{equation}
\begin{split}
  \left ( \left (1+2\frac{p}{n} \right ) + \frac{n^{2} -n - 2p}{Kn} \right ) \alpha_{n,2p} & = 
\left ( \sum_{i+j=p-1} \left (1+2\frac{j}{n} \right ) \alpha_{n,2j} \alpha_{2,2i} \right ) \\ 
 &   +  \left ( \alpha_{n-1,2p} - \alpha_{n+1, 2(p-1)} \right ), 
\end{split}\end{equation}
we see that the recursion hypothesis 
$$\forall (m,2q) \ll (n,2p) \, \; \textnormal{with } \; q\geq 1, \; \alpha_{m,2q} \underset{K \rightarrow + \infty}{\longrightarrow} 0$$ implies 
$\alpha_{n,2p} \underset{K \rightarrow + \infty}{\longrightarrow} 0$ and so this is true for all $p \geq 2$ and $n \geq 2$ by induction principle.
\end{proof}

\par In the no-diffusion case of equation (\ref{eq}), the graph graph $\mathcal H_{OA} = \{ q_{OA}(x), \; x \in [0,1[ \} \subset L^2$,
 where $q_{OA}(x) = \frac{1}{2\pi} + \frac{1}{\pi} \sum_{1}^{+\infty} x^n \cos(n\theta) \;  (x \in [0,1[), $ 
is an invariant manifold for the dynamics (see \cite{MR2464324}).
In that case, the two-dimensional invariant manifold of Ott and Antonsen is obtained from $q_{OA}$ by rotations 
$$ \mathcal M_{OA} = \{ \frac{1}{2\pi} + \frac{1}{\pi} \sum_{1}^{+\infty} x^n \cos(n(\theta + \alpha)), \; x \in [0,1[,\; \alpha \in [0, 2\pi]  \} \subset L^2.$$
Proposition \ref{prop:cv_form} implies that $\varphi_n(x) \underset{K \rightarrow + \infty}{\longrightarrow} x^n $ 
and that $\Phi_\mathcal H \underset{K \rightarrow + \infty}{\longrightarrow} q_{OA}$ formally. 
\begin{prop}
Assume that for all $K$ large enough, we have $x_1^\prime(t) > 0$ for all $t \in ]- \infty , + \infty[$.
Consider $P_N : \, L^2 \rightarrow \vspan \{ \cos(k \theta), \, k \leq N \}$ the natural $L^2$-orthogonal projection. 
For all $z \in ]0,1[$ and $N < + \infty$, 
we have $$ P_N \, \Phi_\mathcal H(z) \underset{K \rightarrow + \infty}{\longrightarrow}  P_N  \,q_{OA}(z)  
\;\;\; \text{ with convergence in } \mathcal G_a \text{ for any } a>0. $$
\end{prop}
\begin{proof}
Recall that (see section  \ref{sec:equilibria})
$r = \int_{-\pi}^{\pi} \hat q (\theta) \cos(\theta) d\theta = \underset{t \rightarrow + \infty}{\lim} x_1(t) 
\underset{K \rightarrow + \infty}{\longrightarrow} 1.$
For any $K$ large enough, we have $x_1 : \, \mathbb R \rightarrow ]0,r[$ is a bijection. 
For any $a \in ]0,1[$ and $K$ large enough, we have $0<a<r$, and there is a unique $t_{a,K} \in ]- \infty, + \infty[$
such that $\int_{-\pi}^{\pi} q_h (t_{a,K},\theta) \cos(\theta) d\theta = x_1(t_{a,K})= a$.
The previous proposition implies that
$$P_N \, \Phi_\mathcal H(z) =  \frac{1}{2\pi} + \frac 1\pi \sum_{n = 1}^N \left( \sum_{0\leq n+2p \leq N} \alpha_{n,2p} z^{n+2p} \right) \cos(n \cdot) \underset{K\rightarrow + \infty}{\longrightarrow}
\frac{1}{2\pi} + \frac 1\pi \sum_{n = 1}^N z^n \cos(n \cdot), $$
uniformly on $[0, 2 \pi]$. 
The degree of $P_N \, \Phi_\mathcal H(z)$ (trigonometric polynomial) is independent of $K$, and so convergence occurs in any $\mathcal G_a$, $a>0$. 
\end{proof}

\par We recall now a classical property of analytic functions. 
\begin{prop}
Let $f_K$ ($K \in \mathbb N$) be a sequence of analytic functions with radius $R_K > \alpha$ for some $\alpha >0$ independent of $K$.
Suppose that $g$ is analytic with radius $R_g > \alpha$ and 
$ f_K \underset{K \rightarrow + \infty}{\longrightarrow} g \;\; \text{ uniformly on } \; [0,\alpha].$
Then we have $ f_K \rightarrow g $ in analytic function space, and in particular for any $l \geq 0$, 
$\frac{d^l}{dx^l} f_K \rightarrow \frac{d^l}{dx^l} g \;\; \text{ uniformly on } \; [0,\alpha].$
\end{prop}

{
As explained above, figures \ref{Fig2} and \ref{Fig4} suggest that $\varphi_n$ are analytical for any $n$, 
and their Taylor series $\sum_p \alpha_{n,2p} x^{2p}$ have the same radius $R = r \underset{K \rightarrow + \infty}{\longrightarrow} 1$.
\begin{num_res}
\par For any $n$ we have $\varphi_n(x) \underset{K \rightarrow + \infty}{\longrightarrow} x^n$ uniformly on any compact subset of $[0,1[$
(see figure \ref{Fig_OttAntonsen}).
\end{num_res}
}
Figure \ref{Fig_OttAntonsen} shows convergence $\varphi_n \rightarrow x^n$ uniformly on any compact of $[0,1[$ when $K \rightarrow + \infty$. 
The difference $\varphi_n - x_n$ is increasing on $[0,r]$ (for any $K$ and $n$), and at $x=r$ we have 
$ \underset{K \rightarrow + \infty}{\lim} \varphi_n(r) - r^n = 0, $
which shows that the convergence is uniform. 
Furthermore we expect theoretically  $\vert \varphi_n(r) - r^n = 0 \vert = \frac{c}{K} + o(\frac 1K)$ when $K \rightarrow 0$, 
and this is confirmed numerically on figure \ref{Fig_OttAntonsen} too, suggesting that the convergence happens with speed $\frac 1K$.

\begin{figure}[h!tbp]
\begin{center}
\psfrag{(a)}[B][B][1][0]{ \small{(a)}}
\psfrag{(b)}[B][B][1][0]{ \small{(b)}}
\psfrag{f2N(x)}[B][B][1][0]{ \small{$\varphi_2(x)$}}
\psfrag{phi2(x)}[B][B][1][0]{ \small{$\varphi_2(x)$}}
\psfrag{x}[B][B][1][0]{ \small{$x$}}
\includegraphics[scale=0.92]{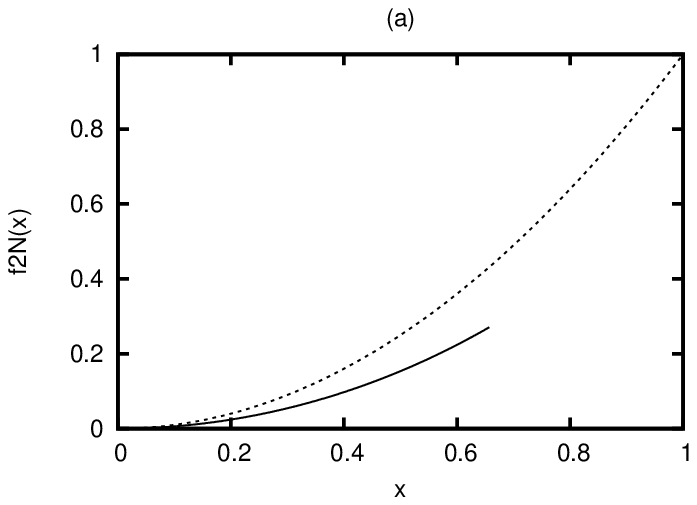}
\includegraphics[scale=0.92]{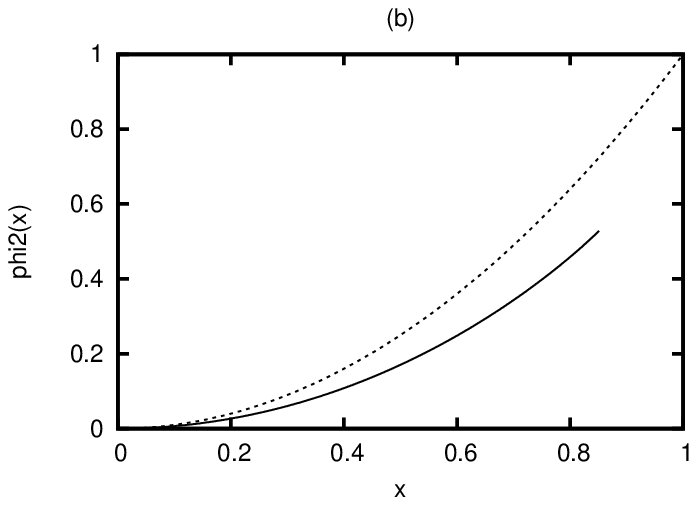} 
\end{center}
\begin{center}
\psfrag{(c)}[B][B][1][0]{ {(c)}}
\psfrag{(d)}[B][B][1][0]{ {(d)}}
\psfrag{phi2(x)}[B][B][1][0]{ \small{$\varphi_2(x)$}}
\psfrag{x}[B][B][1][0]{ \small{$x$}}
\includegraphics[scale=0.92]{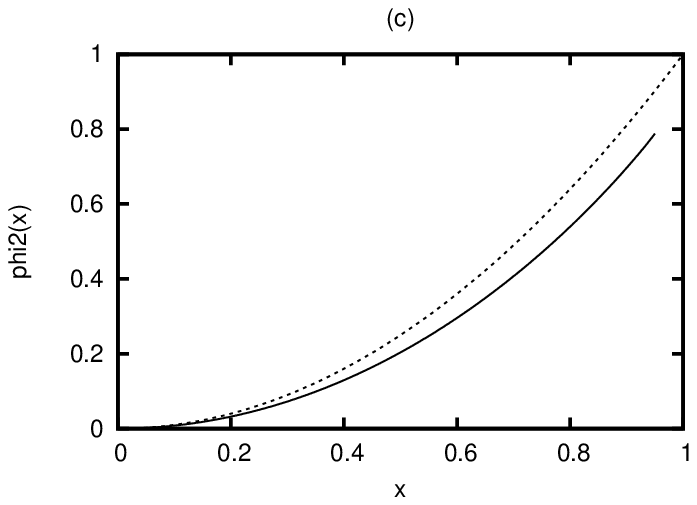}
\includegraphics[scale=0.92]{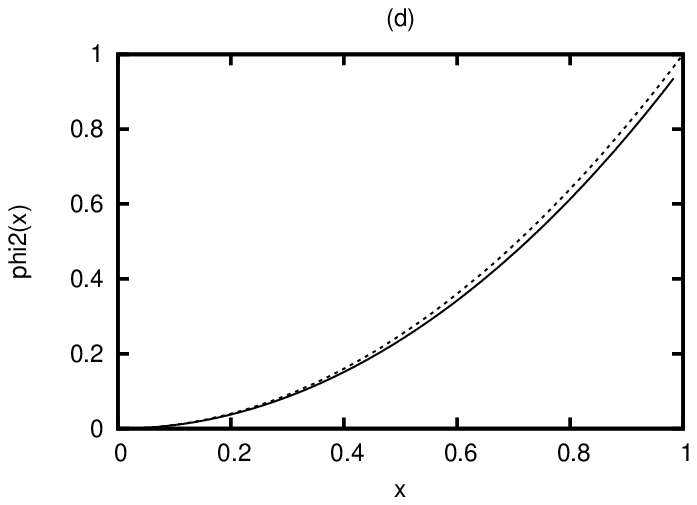} 
\end{center}
\vskip -2mm
\caption{Uniform convergence $\varphi_n (x)\stackrel{K \rightarrow + \infty}{\longrightarrow} x^n$.
\small{Plots of functions $\phi_2^N(x)$ for $x \in [0,r]$, for different values of $K$ :
$K=1.5$ (a), $K=2.0$ (b), $K=4.0$ (c), and $K=16.0$ (d). All values were computed with $N=100$ terms in the entire series defining $\varphi_2$. 
The top most curve in all plots is $y=x^2$. 
The functions $\varphi_2$ converges to the limit $y = x^2$ when $K \rightarrow + \infty$, uniformly on any closed set of $[0,1[$. 
The difference $x^n -\varphi_n(x) $ is increasing for $x \in [0,r[$, so that 
$\sup_{0 \leq x \leq r}  \vert x^n - \varphi_n(x) \vert = r^n - \varphi_n(r)
= \left ( \int \hat{q}(\theta) \cos(\theta) d \theta \right )^{n} - \int \hat{q}(\theta)\cos(n \theta) d \theta $.
One has $r = 1 - \frac{1}{4K} + \mathcal{O} \left ( \frac{1}{K^2} \right )$ and 
$\varphi_n(r) = \frac{I_n(2Kr)}{2 \pi I_0(2Kr)} = 1 - \frac{n^2}{4K} + \mathcal{O}\left ( \frac{1}{K^2} \right )$.
Thus one expects a convergence towards the Ott and Antonsen two dimensional invariant manifold at speed $\frac{1}{K}$ : 
$\sup_{0 \leq x \leq r}  \vert x^n - \varphi_n(x) \vert \leq \frac{C_n}{K}$}
}
\label{Fig_OttAntonsen}
\end{figure}

\par Assuming that these numerical observations are true, we have 
$$ \Phi_\mathcal H (z) = 
 \frac{1}{2\pi} + \frac 1\pi \sum_{n = 1}^N \left( z^n \sum_{0\leq n+2p \leq + \infty} \alpha_{n,2p} z^{2p} \right) \cos(n \cdot), $$
with $\sum_{0\leq n+2p \leq + \infty} \alpha_{n,2p} z^{2p} \underset{K \rightarrow + \infty}{\longrightarrow} 1 $ uniformly on any compact subset of $[0,1[$
 (for any $n$), and $ \vert \sum_{0\leq n+2p \leq + \infty} \alpha_{n,2p} z^{2p} \vert \leq 1 $ for all $n$, $K$ and $z \in [0,r]$.
This would imply $\vert \varphi_n(z) \vert \leq z^n $ (for all $n\geq 1$ and $z \in [0,r]$), and 
$\Phi_\mathcal H (z) = 
 \frac{1}{2\pi} + \frac 1\pi \sum_{n = 1}^N \varphi_n(z) \cos(n \cdot)$ is in Gevrey space 
$\mathcal G_a$ uniformly on $K$, for all $z \leq \alpha$ and $a < \frac 1\alpha$.
\begin{conj}
\par For any $\alpha \in [0,1[$ and any $z \in [0, \alpha]$,  we have 
$$ \Phi_\mathcal H (z) \underset{K \rightarrow + \infty}{\longrightarrow} q_{OA}(z) \;\; \text{ in } \; \mathcal G_a \;\; \forall\, 1<a< \frac 1\alpha. $$
\par The two dimensional global attractor $W^u(\frac{1}{2\pi})$ of (\ref{eq}) converges to the two dimensional invariant manifold 
$\mathcal M_{OA}$ of Kuramoto equation with no diffusion,
in analytic functions space, with speed $\frac 1K$ when $K \rightarrow + \infty$.
\end{conj}

\par The dynamics on $\mathcal H$ are given by 
$x_1^\prime  = - \frac 12 x_1 + K x_1 ( 1 - \varphi_2(x_1)$. 
Considering $\tilde x_1 (t) = x_1(\frac tK)$, in the limit $K \rightarrow + \infty$, 
we find $\tilde x_1^\prime =  \tilde x_1 (1 - \tilde x_1 ^2)$, which is the equation ruling the dynamics on $\mathcal M_{0A}$ is the no diffusion case. 

\begin{conj}
Appropriately time-rescaled by a factor $\frac 1K$, the solutions of equation (\ref{eq}) on $W^u(\frac 1{2\pi})$ converges to solutions 
of the no-diffusion case on $\mathcal M_{0A}$, in Gevrey space on all intervals $]- \infty, T]$, with speed $\frac 1K$.
\end{conj}


\section*{Acknowledgments}
This work has been partially supported by the ANR grant MANDy. 
G. G. acknowledges also the support of ANR grant SHEPI.

\bibliographystyle{amsplain}
\bibliography{GlobalAttractorSub2.bib}

\section*{Appendix}

\subsection*{Proof of lemma \ref{Gevrey_Cv_Lemma} }
\par We will use the classical fact that for any positive time $\eta > 0$, the semiflow $S_\eta$ of equation (\ref{eq}) 
is well defined $L^2 : \, \rightarrow L^2$ and Lipschitz continuous on bounded sets, and the Lipschitz constant can 
be chosen uniformly for $\eta$ in bounded intervals $[0, \epsilon]$ (see \cite{MR610244} for example).

\par Consider $q_0$ and $\tilde{q_0}$ in $L^2$, and $q(t)$ and $\tilde{q}(t)$ the associated solutions of equation (\ref{eq}).
The equation for the difference $w = q - \tilde{q}$ is linear :
\begin{equation} 
\partial_t w = \frac12 \partial_{\theta}^2 w - K \partial_{\theta}( F w ) - K\partial_{\theta}( (J \ast w) \tilde{q} ) 
\end{equation}
where $F (t)=J \ast q (t) = x_1(t) \cos(\theta) - y_1(t) \sin(\theta)$.

Noting $w = \frac{1}{\pi} \sum \left ( \alpha_n \cos(n \theta) + \beta_n \sin(n \theta) \right )$, $\gamma_n = \alpha_n + i \beta_n$, 
$z_1 = x_1 + i y_1$ and $\tilde{z}_n = \tilde{x}_n + i \tilde{y}_n$, 
the previous equation implies the Fourier system :
\begin{equation} 
\gamma_n^{\prime} = -\frac{n^2}{2} \gamma_n + \frac{Kn}{2} ( z_1 \gamma_{n-1} - \overline{z_1} \gamma_{n+1} ) +  
\frac{Kn}{2} ( \gamma_1 \tilde{z}_{n-1} - \overline{\gamma_1} \tilde{z}_{n+1} ) 
\end{equation}
with $\gamma_0= 0$ for all $t\geq 0$.
And taking $A_n = z_1 \overline{\gamma_n} \gamma_{n-1} + \overline{z_1} \gamma_n \overline{\gamma_{n-1}}$
and $B_n = \gamma_1 \overline{\gamma_n} \tilde{z}_{n-1} +  \overline{\gamma_1} \gamma_n \overline{\tilde{z}_{n-1}} $
we have (similarly to the previous proof)
\begin{equation} 
\frac{d}{dt} \vert \gamma_n \vert^2 = - n^2 \vert \gamma_n \vert + \frac{Kn}{2}(A_n - A_{n+1}) + \frac{Kn}{2}(B_n - B_{n+1}).  
\end{equation}

Notice that for any time $\eta \geq 0$ and for any $n\geq 1$, 
\begin{equation}
\vert \gamma_n(\eta) \vert \leq \Vert w(\eta) \Vert_{L^2} = \Vert q(\eta) - \tilde{q}(\eta) \Vert_{L^2}  \leq Lip(S_\eta) \Vert q_0 - \tilde{q}_0 \Vert_{L^2},
\end{equation}
so for $f(n)=a^{n \min(t,1)}$, using $ \vert z_1 \vert \leq 2 $, we have 
\begin{equation} 
\frac{d}{dt} \left ( \sum_{n = N}^M \frac{f(n)}{n} \vert \gamma_n \vert^2 \right ) +
\left ( 1-\frac{2\log(a)}{N} - \frac{C}{\sqrt{N(N+1)}} \frac{a^2 -1}{a} \right )  \left ( \sum_{n = N}^M nf(n) \vert \gamma_n \vert^2 \right ) 
\end{equation}
\begin{equation} 
\leq C f(N) \vert \gamma_N \vert \vert \gamma_{N-1} \vert + C \vert \gamma_1 \vert \sum_{n = N+1}^M f(n) \vert \gamma_n \vert \vert \tilde{z}_{n-1} \vert
\end{equation}
\begin{equation}
\leq C(t) \Vert q_0 - \tilde{q}_0 \Vert_{L^2}^2 + \frac{C}{N(N+1)} \left ( \sum_{n = N}^M nf(n) \vert \gamma_n \vert^2  \right ) 
  + C \vert \gamma_1 \vert^2 \left ( \sum_{n = N}^M nf(n) \vert \tilde{z}_n \vert^2 \right )
\end{equation}
\begin{equation}
\leq C(t) \Vert q_0 - \tilde{q}_0 \Vert_{L^2}^2 + \frac{C}{N(N+1)} \left ( \sum_{n = N}^M nf(n) \vert \gamma_n \vert^2  \right ) 
+ C(t) \Vert q_0 - \tilde{q}_0 \Vert_{L^2}^2 \Vert \tilde{q} \Vert_{a^{\min(1,t)},N,M}^2
\end{equation}

From theorem \ref{th:G1}, we have that $\vert \tilde{q}(t) \vert_{G_{a^{\min(1,t)}}}$ is bounded on $[0, \epsilon]$,
 for any $a \geq 1 $, and so $\Vert \tilde{q}(t) \Vert_{G_{a^{\min(1,t)}}}$ is bounded on $[0, \epsilon]$ too (for any $a \geq 1$).
The constants $C$ above are all independent of $M$, 
so there is a $N$ large enough such that for all $t \in [0, \epsilon]$,
\begin{equation} 
2 \frac{d}{dt} \left ( \sum_{n \geq N} \frac{f(n)}{n} \vert \gamma_n \vert^2 \right )^{1/2} +
C  \left ( \sum_{n \geq N} nf(n) \vert \gamma_n \vert^2 \right )^{1/2} \leq C(t) \Vert q_1(0) - q_2(0) \Vert_{L^2}. 
\end{equation}

This proves that for any $q_0$ in $L^2$, for any $\epsilon >0$ and $\delta >0$, for any $a \geq 1$, there is a $M < + \infty$ such that for 
any $\tilde{q}_0 \in B_{L^2}(q_0, \epsilon)$ we have 
$\Vert q(\delta) - \tilde{q}(\delta) \Vert_{G_{a^{min(1,\delta)}}} \leq M \Vert q_0 - \tilde{q}_0 \Vert_{L^2}$.

\subsection*{Proof of theorem \ref{thm_PIM}}

\textit{Definition of the transform $\mathcal F$ and the new equation (\ref{eq2}).}
In the following we slightly abuse notations by using the same letter $V$ to denote both 
\begin{equation}
\begin{array}{rl}
V &= V(q)(\theta) = -K \left ( x_1(q) \cos(\theta) + y_1(q) \sin(\theta)  \right ),  \\
\;\; \text{ or } \; \;\; V &= V(a,b)(\theta) = -K \left ( a \cos(\theta) + b \sin(\theta)  \right )
\end{array}
\end{equation} depending on context. 
For $q(t)$ a solution of (\ref{eq}), the ordinary differential equations satisfied by $x_{1} = x_{1}(q(t))$ and $y_{1}=y_{1}(q(t))$ are 
$ x_{1}^{\prime} = - \frac{1}{2}x_{1} + \frac{K}{2} \left[ x_{1}(1-x_{2}) - y_{1}y_{2} \right] $ and 
 $y_{1}^{\prime} = - \frac{1}{2}y_{1} + \frac{K}{2} \left[ -x_{1}y_{2} + y_{1}(1+x_{2}) \right] $ 
and this gives in particular 
\begin{equation}
\partial_{t} V = - \frac{V}{2} - \frac{K}{2} \left[ F_{1}(u,x_{1},y_{1})\cos(\theta) + G_{1}(u,x_{1},y_{1}) \sin(\theta) \right]. 
\end{equation} 
Along a solution of the Kuramoto equation (\ref{eq}), 
where $u = e^V q$ and
\begin{equation} 
\begin{array}{rl}
F_{1}(u, a, b) &= \left[ a(1-X_{2}(u,a,b)) - bY_{2}(u,a,b) \right] \\
G_{1}(u, a, b) &= \left[ -aY_{2}(u,a,b) + b(1+X_{2}(u,a,b)) \right], 
\end{array}
\end{equation}
with
\begin{equation} 
\begin{array}{rl}
X_{2}(u, a, b) &= \int_{- \pi}^{\pi} u(\theta) e^{\frac{K}{2}(a\cos(\theta) + b \sin(\theta))}\cos(2 \theta) d\theta, \\
Y_{2}(u, a, b) &= \int_{- \pi}^{\pi} u(\theta) e^{\frac{K}{2}(a\cos(\theta) + b \sin(\theta))}\sin(2 \theta) d\theta. 
\end{array}
\end{equation}
We have that for any solution $q$ of equation (\ref{eq}), $\mathcal{U}$ satisfies (\ref{eq2})
where
\begin{equation} A \mathcal{U} = 
\left(  \begin{array}{ccc} 
- \frac{1}{2} \Delta & 0 & 0 \\ 
0 & \frac12 & 0 \\
0 & 0 & \frac12 
\end{array}  \right)
\left( \begin{array}{c} u \\ x_{1} \\ y_{1} \end{array} \right) = 
\left( \begin{array}{c} - \frac{1}{2} \partial_{\theta}^2 u \\ \frac12 x_{1} \\ \frac12 y_{1} \end{array} \right),   
\end{equation}
\begin{equation} N(\mathcal{U}) = 
\left( \begin{array}{c}
-\frac{K}{2} \left[ \frac12 V + \frac{K}{2}(\partial_{\theta}V)^{2} + \frac{K}{2} F_{1}(\mathcal{U})\cos(\theta) 
+ \frac{K}{2} G_{1}(\mathcal{U}) \sin(\theta)  \right]u \\
\frac12 K F_{1}(\mathcal{U}) \\
\frac12 K G_{1}(\mathcal{U})
\end{array}  \right) \end{equation}\\

We will show now that inertial manifolds exist for equation (\ref{eq2}).
There are several inertial manifolds existence theorem (see \cite{MR2511702,MR2725294,MR1409653} for instance), such as the following. 
\begin{thm}\label{thm_InertialManifolds}
Consider an abstract evolution equation on a Hilbert space $H$
\begin{equation}\label{eq_abstact} \frac{du}{dt} +Au = f(u) \end{equation}
where $A$ is a positive self-adjoint operator with compact resolvent and $f : \, H \rightarrow H$ is Lipschitz continuous on bounded sets.
The spectrum of $A$ consists in a sequence of eigenvalues $\lambda_n$ such that $\lambda_n \leq \lambda_{n+1} \underset{n \rightarrow \infty}{ \longrightarrow} + \infty$. 
Suppose that equation (\ref{eq_abstact}) has a bounded absorbing set in $H$. Then there exists a constant $C$, such that if 
\begin{equation}\label{eq_SGC} \lambda_{n+1} - \lambda_n > C \end{equation} for some $n \geq 0$, 
then there exists an asymptotically complete inertial manifold $\mathcal M^*$ for equation (\ref{eq_abstact}), 
which is the graph of a Lipschitz function $\Phi$ on $E_n = \oplus_{i \leq n} \ker( A - \lambda_i Id )$.
Furthermore, if $f$ is $C^1(H,H)$, then the graph function $\Phi$ and the inertial manifold $\mathcal M^*$ are $C^1$ too.
\end{thm}
During the next two steps of this proof, we show that theorem \ref{thm_InertialManifolds} applies to equation (\ref{eq2}).\\

\par \textit{The Cauchy problem for equation (\ref{eq2}) is well-posed.} 
The operator $A$ is diagonal in the natural Hilbert basis of $E := H^{s} \times \mathbb{R} \times \mathbb{R}$,
and the spectral gap condition of theorem \ref{thm_InertialManifolds} holds for the operator $A$.
The maps $X_{2}, Y_{2}, F_1, \text{ and } G_1 \, : \, L^{2} \times \mathbb{R} \times \mathbb{R} \rightarrow \mathbb{R}$ are $C^{\infty}$. 
Since the first line of the non linear term $N ( \mathcal{U})$ is the product of a Fourier polynomial 
(whose coefficients are $C^{\infty}(H^{s} \times \mathbb{R} \times \mathbb{R})$) and $u$, we have that $N$ is 
$C^{\infty} : \; H^{s} \times \mathbb{R}^2 \rightarrow H^{s} \times \mathbb{R}^2 $. 
In particular, the non linear term $N$ is Lipschitz on bounded subsets of any 
$H^{s} \times \mathbb{R} \times \mathbb{R}$, and  classical semi-linear parabolic equations results in Hilbert spaces apply 
(see for example \cite{MR1691574}).
\begin{prop} Let $s \in \mathbb{N}$. 
For any $\mathcal{U}_{0} \in H^{s+2} \times \mathbb{R} \times \mathbb{R}$, 
there is a unique solution $\mathcal{U}(t)$ of (\ref{eq2}) with 
$\mathcal{U} \in C^{0}([0, T_{max}[; H^{s+2} \times \mathbb{R} \times \mathbb{R} ) \cap 
C^{1}([0, T_{max}[; H^{s} \times \mathbb{R} \times \mathbb{R} )$, 
and we have either $T_{max} = + \infty$ or 
$\underset{t \rightarrow T_{max}}{\lim} \Vert \mathcal{U}(t) \Vert_{H^{s} \times \mathbb{R} \times \mathbb{R}} = + \infty$. \\
Furthermore, if $ \mathcal{U}_{0} \geq 0 $, then we have $\mathcal{U}(t) \geq 0$  for all $t \in [0 , T_{max}[$.
\end{prop}
\par In the following we denote by $S^*_t$ the semiflow generated by equation (\ref{eq2}).
Thanks to the regularizing properties of equation (\ref{eq}), without loss of generality we may assume that $q_{0} \in H^{s}$ 
and $\mathcal{U}_{0} \in H^{s} \times \mathbb{R} \times \mathbb{R}$ (for any fixed $s \geq 0$). \\

\textit{Existence of a bounded absorbing set for equation (\ref{eq2}).} 
We have already established that all hypotheses of theorem \ref{thm_InertialManifolds} hold for equation (\ref{eq2}), 
but the existence of a bounded absorbing set.
\par For any $\mathcal U$ solution of (\ref{eq2}),  $\left ( q,a,b \right ) = \mathcal F^{-1} \left( \mathcal U \right) $ satisfies 
\begin{equation}\label{eq2q}
 \partial_{t} q = \frac{1}{2} \partial_{\theta}^2 q + K \partial_{\theta} \left [ (- a \sin(\cdot) + b \cos(\cdot)) q \right ],  
\end{equation}
\begin{equation}\label{eq_ab}
\begin{array}{ccc}
a^{\prime} &=& -\frac12 a + \frac12 K F_{1}(e^{-K(a \cos (\cdot) + b \sin (\cdot) )}q, a, b),  \\
b^{\prime} &=& -\frac12 b + \frac12 K G_{1}(e^{-K(a \cos (\cdot) + b \sin (\cdot) )}q, a, b).
\end{array}
\end{equation}
As for Kuramoto equation (\ref{eq}), we prove that equations (\ref{eq2q}) and (\ref{eq_ab}) 
have a bounded absorbing set in a Gevrey space, by mean of an equivalent Fourier ODE system.
\par Denoting $w_{1}(t) = a(t) + i b(t) \in \mathbb{C}$, 
and $z_{n}(t) = x_{n}(t) + i y_{n}(t) = \int_{- \pi}^{\pi} e^{in\theta} q(t,\theta) d \theta$
we have
\begin{equation} z_{n}^{\prime} = - \frac{n^{2}}{2} z_{n} + \frac{Kn}{2}[ w_{1} z_{n-1} - \overline{w_{1}} z_{n+1}  ]. 
\end{equation} 

\par With the method of theorem \ref{th:G1} and lemma \ref{Gevrey_Cv_Lemma}, we obtain
\begin{equation} \frac{d}{dt} \vert q \vert_{f,N}^{2}   
- \sum_{N}^{+ \infty} \frac{f^{\prime}(n,t)}{n} \vert z_{n} \vert^{2} 
+2 \Vert q \Vert_{f,N}^{2} \leq 2K f(N,t) + \frac{2 K}{\sqrt{N(N-1)}} \frac{(a^{2}-1)}{a} \vert w_{1}(t) \vert \Vert q \Vert_{t,N}^{2}. 
\end{equation}
If $w_{1}$ is bounded on some time interval $[0,T]$, then for $N$ large enough we have
 $\frac{d}{dt} \vert q \vert_{f,N}^{2} + \Vert q \Vert_{f,N}^{2} \leq C, $
and for any solution $\mathcal{U} \in C^{1}([0,T], L^{2} \times \mathbb{R} \times \mathbb{R})$, we have
$u(t) \in H_{f}$ for all $t \in ]0; T]$. We also see from the previous argument that
 if we have a uniform bound $\vert w_{1} \vert \leq C_{1}$  for all $t \in [0 , + \infty[$, then $N$ can be chosen large enough such that 
$\frac{d}{dt} \vert q \vert_{f,N}^{2} + \Vert q \Vert_{f,N}^{2} \leq C $
and the bounded ball $B_{H_{f}}(0, 2 C)$ in Gevrey space is an absorbing set for our system. \\

\par We now prove a uniform bound for $w_{1}$ on $[0 , + \infty[$. 
Equation (\ref{eq2q}) preserves the integral of $q$,  
$\int_{- \pi}^{\pi} q(t) d \theta = \int_{- \pi}^{\pi} q_{0} d \theta$ along solutions. 
Since $q_{0} \geq 0$ implies $q(t) \geq 0$ for solutions of equations (\ref{eq2q}) and (\ref{eq_ab}), 
we have $\vert x_{1} \vert \leq \int_{-\pi}^{\pi} q_0(\theta) d \theta= C$ and $\vert y_{1} \vert \leq C$ as long as the solution exists
(where $x_1 = \int_{-\pi}^\pi q(t,\theta) \cos(\theta) d\theta$ and $y_1 = \int_{-\pi}^\pi q(t,\theta) \sin(\theta) d\theta$).
Consider $(q_0, a_0, b_0) \in L^2 \times \mathbb R \times \mathbb R$ and the associated solution $ ( q(t),a(t),b(t) ) $ 
of equations (\ref{eq2q}) and (\ref{eq_ab}).
We have
\begin{equation}\label{eq_x1y1}
\begin{array}{ccc}
x_1^{\prime} &=& -\frac12 x_1 + \frac12 K F_{1}(e^{-K(a \cos (\cdot) + b \sin (\cdot) )}q, a, b),  \\
y_1^{\prime} &=& -\frac12 y_1 + \frac12 K G_{1}(e^{-K(a \cos (\cdot) + b \sin (\cdot) )}q, a, b), 
\end{array}
\end{equation}
\begin{equation}\label{eq:x1a_y1b}
\text{ so that } \;\;\; \frac{d}{dt} (x_1 - a) = - (x_1 -a) \;\;\; \text{and} \;\;\; \frac{d}{dt} (y_1 - b) = - (y_1 -b).
\end{equation}
The quantities $\vert x_1 -a \vert $ and $\vert y_1 - b \vert$ are exponentially decreasing. 
Since $\vert x_1(t) \vert \leq 1$ and $\vert y_1(t) \vert \leq 1$ for all time $t \geq 0$, 
there is a time $t_0$ such that $\vert w_1 \vert \leq \vert a \vert + \vert b \vert \leq 4$ for all $t \geq t_0$.
As seen in the previous paragraph, this implies that 
\begin{equation} \frac{d}{dt} \vert q \vert_{f,N}^{2} + \Vert q \Vert_{f,N}^{2} \leq C_{1} 
\;\;\; \forall t \in [0, + \infty[,
\end{equation} 
and the bounded ball $B_{H_{f}}(0, 2C_{1})$ in Gevrey space is an absorbing set for our equations (\ref{eq2q}) and (\ref{eq_ab})
and so for equation (\ref{eq2}). \\

\textit{Existence of an AcIM for equation (\ref{eq2})  and an AcSIM for equation (\ref{eq}).}
By theorem \ref{thm_InertialManifolds}, there exists an inertial manifold $\mathcal M^*$ for equation (\ref{eq2}).
Furthermore, since the non-linear term $\mathcal N \left ( \mathcal U \right )$ is $C^k$ ($k\geq1$), 
the inertial manifold $\mathcal M^* $ is $C^1$, and
for any $\mathcal U_0 \in H^s \times \mathbb R^2 $ and $t_0$ large enough 
there is a unique phase $\mathcal V_0 \in \mathcal M^*$ such that 
\begin{equation}
\left \Vert S_{t+t_0}^* \mathcal U_0 - S_t^* \mathcal V_0 \right \Vert = \mathcal O \left ( e^{- \eta t} \right )
\end{equation}
and the leaves 
\begin{equation}
\mathcal L_{\mathcal V_0} = \left  \{  \mathcal U_0, \;\; 
 \left \Vert S_{t}^* \mathcal U_0 - S_t^* \mathcal V_0 \right \Vert = \mathcal O \left ( e^{- \eta t} \right )    \right  \} 
\end{equation}
form a $C^1$ foliation of space $H^s \times \mathbb R^2 $.\\

\par To deduce properties of the flow of equation (\ref{eq}) from results for solutions of equation (\ref{eq2}), 
it is essential that the transform  $\mathcal{F} \circ I \, : \, q \mapsto \mathcal{U}$  and the inverse 
$\mathcal{F}^{-1} $ are smooth. 
With the definition $u(\theta) = q(\theta) e^{\frac12 K V(x_{1}, y_{1}, \theta)}$, one sees that 
\begin{equation} \Vert u \Vert_{L^{2}} \leq C_{K} \Vert q \Vert_{L^{2}} \Vert e^{- \frac12 K(x_{1}\cos(\cdot) + y_{1} \sin(\cdot) )} \Vert_{L^{\infty}} 
\leq C_{K} C(K, x_1, y_1 \Vert q \Vert_{L^{2}} ) \end{equation} 
and also for any $s \in \mathbb{N}$
$ \Vert u \Vert_{H^{s}} \leq  C_{K,s} C(K, s, x_1, y_1) \Vert q \Vert_{H^{s}} , $ 
so that $q \mapsto u$ is well defined $H^{s} \rightarrow H^{s}$.
Besides the mappings $q \mapsto (x_{1}, y_{1}) \mapsto V(x_{1}, y_{1}, \cdot) \mapsto e^{\frac{K}{2} V} $ are smooth 
and one finally finds that $ F \, : \, q \mapsto \mathcal{U}$ is $C^{\infty}(H^{s} ; H^{s} \times \mathbb{R} \times \mathbb{R})$ 
for any $s \in \mathbb{N}$. 
For the same reasons, the inverse mapping $\mathcal F^{-1}$ 
$ \mathcal U=(u,a,b) \mapsto (q,a,b) = (u e^{\frac 12 K (a \cos(\cdot) + b \sin(\cdot)) }) $
is also well defined and $C^{\infty}$ smooth 
$H^{s} \times \mathbb{R} \times \mathbb{R} \rightarrow H^{s} \times \mathbb{R} \times \mathbb{R} $
for any $s \geq 0$. \\

\par For any solution $q$ of equation (\ref{eq}), $\mathcal U = \mathcal F \circ I (q)$ is a solution of (\ref{eq2}), 
and $\{ (q,a,b) \in H^s \times \mathbb R^2, \; x_1(q) =a, y_1(q)=b \}$ is an invariant manifold for equation (\ref{eq2}).
To any $q_{0} \in H^{s}$, we associate 
$\mathcal{U}_{0} = \mathcal{F}(q_{0})$, which has a phase on the inertial manifold $\mathcal M^*$ of equations (\ref{eq2q}) and (\ref{eq_ab}), 
 ie there is a $t_{0}$ and a $\mathcal{V}_{0} \in \mathcal M_n^*$ such that 
$$ \Vert S^{*}(t+t_{0})\mathcal{U}_{0} - S^{*}(t) \mathcal{V}_{0} \Vert_{\mathcal{E}} = \mathcal{O}\left( e^{- \eta t} \right). $$
Consider the projection $P : H^s \times \mathbb R^2 \rightarrow H^s$ defined by $P (q,a,b) = q$.
Taking $v(t) = P \mathcal{F}^{-1} S^*_t \mathcal{V}_{0}$, since the projection $P$ and the inverse $\mathcal F^{-1}$ are smooth, we have
$$ \Vert S(t+t_{0})q_{0} - v(t) \Vert_{H^{s}} = \mathcal{O}\left( e^{- \eta t} \right). $$ 
The inertial manifold $\mathcal M^*$ is finite dimensional, so that $\tilde S_t =  S^*_t \vert_{\mathcal M^*}$ is a flow on a finite dimensional space. 
Since $\mathcal M^*$ is a smooth graph on $E_n$ (see theorem \ref{thm_InertialManifolds}), 
$\mathcal M = P \mathcal F^{-1} \mathcal M^*$ is a smooth graph too. 
We conclude that $\mathcal M = P \mathcal F^{-1} \mathcal M^*$ is an AcSIM for equation (\ref{eq}).

\subsection*{Proof of lemma \ref{lm:new_local_coordinates}}

%
%
To simplify the notation we will write $H^{-1}$
 as as a short-cut for $H^{-1}_{1/\hat q}$ in this proof.
Recall that $L_{\hat q}$ is normal, and we have the orthogonal sum $ L^2 = D(L_{\hat q}^{\frac 12}) = \ker(L_{\hat{q}}) \oplus R(L_{\hat{q}}) $.
 
Let us denote by $P$ the orthonormal projection onto $\ker(L_{\hat{q}})$. 
For $u$ and $x_\varphi$ in a neighborhoods of the origins, 
the equality $u = x_{\varphi} + y$ with $y \in R(L_{\hat q})$ is equivalent to $Pu = Px_{\varphi}$ and $y = z_{u} - z_{\varphi}$.
We check that for each $u$ in a neighborhood of $0$ there is a unique $x_{\varphi(u)}$ 
such that $Pu = Px_{\varphi(u)}$, and variables change $u \mapsto (x_{\varphi(u)}, z_{u} - z_{\varphi(u)})$ will be well defined.
 
Since $x_{\varphi} = \hat{q}( \cdot +\varphi) - \hat{q}( \cdot) $, 
and $\frac{dx_{\varphi}}{ d \varphi}\big \vert _{\varphi = 0} = \partial_{\theta} \hat{q} \neq 0$,
the set $\{ x_{\varphi}: \; \varphi \in ]-\epsilon, \epsilon[ \}$ is a graph above 
$\mathbb{R}\cdot\partial_{\theta} \hat{q} = \ker(L_{\hat{q}})$ for small enough $\epsilon$. 
In particular that for $u$ close enough to zero, there is a unique $\varphi(u)$ such that $Pu = Px_{\varphi(u)}$.
Moreover, $\varphi \mapsto x_{\varphi}$ is $C^{k}( ]- \epsilon, \epsilon[, H^{-1}_{1/\hat{q}} )$ ($k \geq 0$) 
and the decomposition $u \mapsto (x_{\varphi(u)}, y = z_{u} - z_{\varphi(u)})$ 
is $C^{k}$ in a neighborhood of $u = 0$ in $H^{-1}_{1/q}$.
The mapping $u \mapsto (\varphi, y)$ is well defined and $C^{k} \, : \, H^{-1}_{1/q} \rightarrow ]- \epsilon, \epsilon[ \times H^{-1}_{1/q}$ on a neighborhood of the origin. 
%

We can now write the dynamics for the new variables $\varphi$ and $y$.
The $x_{\varphi}$ are equilibria, ie we have $L_{\hat{q}} x_{\varphi} = f( x_{\varphi})$, for $\varphi \in [0 , 2 \pi]$.
If $u$ is a solution of $\frac{du}{dt} + L_{\hat{q}} u = f(u)$, we have
\begin{equation} \left ( \partial_{\varphi} x_{\varphi} \right ) \frac{d \varphi}{dt} + \frac{dy}{dt} + L_{\hat{q}} y = f(x_{\varphi} + y) - f(x_{\varphi})\, .
\end{equation}
Let us denote $A_{2}$ the restriction of $L_{\hat{q}}$ on $R(L_{\hat{q}})$, ie  $A_{2} = L_{\hat{q}} \vert_{R(L_{\hat{q}})}$,
and take $v = {\partial_{\theta}q}/{\Vert \partial_{\theta}q \Vert^{2}}$
 (so that $L_{\hat{q}} v = 0$ and $\ll v,\partial_{\varphi} x_{\varphi }\gg 
 \vert_{\varphi =0} = \ll v, \partial_{\theta} \hat{q}\gg = 1$). 
Then we get the system 
\begin{equation}  \frac{d \varphi}{dt} = \Phi(\varphi , y)  \, \; \textnormal{ and } \, \;
 \frac{dy}{dt} + A_{2} y = g(\varphi , y)\, , \end{equation}
where 
\begin{equation} 
\Phi(\varphi , y) = \frac{1}{\ll v , \partial_{\varphi} x_{\varphi}\gg} \ll v , f(x_{\varphi} + y) - f(x_{\varphi})\gg  \, ,
\end{equation}
and
 \begin{equation}  g(\varphi, y ) =f(x_{\varphi} + y) - f(x_{\varphi}) - \partial_{\varphi}x_{\varphi} \Phi(\varphi, y) \, .
\end{equation} 

We notice that the non linearity $f$ is well defined $L^{2}(\mathbb{S}) \rightarrow H^{-1}(\mathbb{S})$, 
and since $f$ is bilinear, we have $f \in C^{k}(L^{2}, H^{-1} )$ (like before, $k$ is arbitrary). 
Since $f \in C^{k}(L^{2}, H^{-1} )$ , $x_{\varphi} \in C^{k}( [0, 2 \pi] ; L^{2} )$ and 
$\partial_{\varphi} x_{\varphi} \vert_{\varphi = 0} = \partial_{\theta} \hat{q} \neq 0$, 
it is easy to check that there is a $\epsilon > 0$ and an $\eta > 0$ such that 
$ \Phi : ]- \epsilon, \epsilon[ \times B_{L^{2}}(0, \eta) \rightarrow \mathbb{R} $
and $ g : ]-\epsilon, \epsilon[ \times B_{L^{2}}(0, \eta) \rightarrow H^{-1}$ are $C^{k}$. 
The  Frechet derivative of $f$ at the origin is null, ie $Df(0) = 0$, 
and the expressions for $\Phi$ and $g$ give $D \Phi(0,0) = 0$ and $D g(0,0) = 0$. 

\subsection*{Proof of lemma \ref{th:lem4}}

Inequality (\ref{eq:estim_A_2}) also holds if we put $D(A_{2}^{\gamma + \alpha})$ and $D(A_{2}^{\gamma})$ instead of $D(A_{2}^{\alpha})$ and $H^{-1}_{1/q}$. 
For $\alpha = 0$, one directly has $\Vert e^{tA_{2}} v \Vert_{D(A_{2}^{\gamma})} \leq e^{- \beta t} \Vert v \Vert_{D(A_{2}^{\gamma})} $ 
for any $t \geq 0$. We now choose $\alpha=1/2$ and  a  value for  $\epsilon$.
\par Since the Frechet derivative of $\Phi$ and $g$ are smooth, for any $\delta > 0$, there is an $\eta> 0$ such 
that $ \vert \varphi \vert + \Vert y \Vert \leq \eta$ implies 
\begin{equation} 
\Vert D \Phi(\varphi,y) \Vert \leq \delta, \ \ \  \textnormal{ and } \ \ \  \Vert D g(\varphi,y) \Vert \leq \delta\, ,
\end{equation} 
where the norms are the appropriate operator norms.
We also see that for any $\varphi \in ]- \epsilon , \epsilon[$, we have 
$ \Phi(\varphi, 0) = 0$ and $ g(\varphi, 0) = 0$. 
Thus for any $\rho > 0$, there is a $\gamma( \rho ) $ such that on the set 
$\{ (\varphi, y) ; \; \vert \varphi \vert + \Vert y \Vert_{L^{2}} \leq \rho\}$, we have 
\begin{equation} 
\label{eq:gammarho}
\vert \Phi(\varphi, y) \vert + \Vert g(\varphi, y) \Vert_{H^{-1}} \leq \gamma( \rho ) \Vert y \Vert_{L^{2}}, \, \; \textnormal{and } \, \;
\gamma( \rho) \underset{ \rho \rightarrow 0}{\longrightarrow} 0\, . \end{equation}

By the hypothesis on the initial condition and by the continuity of $t \mapsto \vert \varphi_{t} \vert + \Vert y_{t} \Vert_{L^{2}} $
we have that there exists $T>0$ such that $\vert \varphi_{t} \vert + \Vert y_{t} \Vert_{L^{2}} \le \rho$.
We check first that the estimates of the lemma are true on $[0,T]$.

For $y_{0} \in R(L_{\hat{q}})$, we have
\begin{equation} y(t) = e^{-tA_{2}}y_{0} + \int_{0}^{t} e^{-(t-s)A_{2}}g(\varphi(s),y(s))ds
\, ,  \end{equation}
which implies on $[0,T]$
\begin{equation} 
\begin{split}
\Vert y(t) \Vert_{L^{2}} \, &\leq\,  \Vert y_{0} \Vert_{L^{2}} e^{- \vert \lambda _1\vert t} + 
\int_{0}^{t} \Vert e^{-(t-s)A_{2}} \Vert_{\mathcal{L}(H^{-1}, L^{2})} \Vert g(\varphi(s),y(s)) \Vert_{H^{-1}}ds
\\
&\leq\,  \Vert y_{0} \Vert_{L^{2}} e^{- \vert \lambda _1\vert t} +  C_{1/2, \epsilon}
\int_{0}^{t}  \frac{1}{\sqrt{t-s}} e^{- (1-\epsilon) \vert \lambda _1\vert(t-s)} \gamma(\rho) \Vert y(s) \Vert_{L^{2}} ds\, ,
\end{split}  
\end{equation}
where $\epsilon$ is arbitrarily chosen in $]0,1/2[$.

Then we consider $u(t) := \sup_{0 \leq s \leq t} \Vert y(s) \Vert_{L^{2}} e^{(1-2\epsilon)\vert \lambda _1\vert  s}$. 
We have for $t \in [0,T]$
\begin{multline} \Vert y(t) \Vert_{L^{2}} e^{(1-2\epsilon)\vert \lambda _1\vert t} \leq \Vert y_{0} \Vert_{L^{2}} + 
\gamma(\rho) C_{1/2, \epsilon} \int_{0}^{T} \frac{1}{\sqrt{t-s}} e^{- \epsilon \vert \lambda _1\vert (t-s)} u(s) ds 
 \\
 \leq \Vert y_{0} \Vert_{L^{2}} + \gamma(\rho) M_\epsilon u(T)\, ,
 \end{multline}
 where $M_\epsilon := C_{1/2, \epsilon} \int_{0}^{\infty} \frac{1}{\sqrt{s}} e^{- \epsilon \vert \lambda _1\vert (s)}  ds$
and this means (choose $\rho$ so that $\gamma(\rho) \le  1/(2M_\epsilon)$)
\begin{equation} 
u(T) \leq 2 \Vert y_{0} \Vert_{L^{2}}\, ,
\end{equation}
which directly implies \eqref{eq:lem4.2} with $\beta:= (1-2\epsilon) \vert \lambda_1\vert$, still for $t\in [0,T]$.

We now want to relax the condition $t \in [0,T]$. For this we first observe that since
 $\vert \frac{d \varphi}{dt} \vert = \vert \Phi(\varphi,y) \vert \leq \gamma(\rho) \Vert y \Vert_{L^{2}}$, we have
\begin{equation} 
\left\vert \frac{d \varphi}{dt} \right\vert \leq {2\gamma(\rho)} \Vert y_{0} \Vert_{L^{2}} e^{- \beta t} \, ,
\end{equation}
and therefore $ \vert \varphi (t) \vert \leq \vert \varphi_{0} \vert + 
 {\gamma(\rho)}\frac{\rho}{4\beta}  $. If we now choose  
 $\rho_{0}$ such that 
 we have 
$\gamma(\rho) \le   \beta $, then on $[0,T]$ we have
\begin{equation} \vert \varphi (t) \vert \leq \vert \varphi_{0} \vert + \frac{1}{4} \rho \leq \frac{3}{8} \rho\, , 
\end{equation}
and 
\begin{equation} \Vert y(t) \Vert_{L^{2}} + \vert \varphi (t) \vert \leq 
2 \Vert y_{0} \Vert_{L^{2}} + \frac{3}{8} \rho \leq \frac{5}{8} \rho < \rho\, . 
\end{equation}
This implies that there is a $T_{1} > T$ such that $ \Vert y(t) \Vert_{L^{2}} + \vert \varphi (t) \vert \leq \rho $ on $[0, T_{1}]$
and therefore that $T$ can be chosen arbitrarily large. In fact if we set 
 $T_\star:= \sup\{ t:\, \vert \varphi_{t} \vert + \Vert y_{t} \Vert_{L^{2}} \le \rho\}$ and if $T_\star <\infty$
 we can take $T=T_\star$ in the first part of the proof and we arrive to a contradiction.

\subsection*{Proof of corollary \ref{Corollaire_Normal_Hyperbolicity}}
Because of rotation symmetry, it is enough to study the system around $\hat{q}=q_{\varphi = 0}$. 
We consider the splitting $E = L^{2} = \ker(L_{\hat{q}}) \oplus R(L_{\hat{q}}) = T_{\hat{q}} \mathcal{C} \oplus N^{s}_{\hat{q}}$, setting $N^{s}_{\hat{q}} = R(L_{\hat{q}})$.
This splitting extends continuously to all  $\hat q ( \cdot + \varphi ) \in \mathcal C$ by rotations.

Any point of $\mathcal{C}$ is an equilibrium, so that for any time $t$ and integer $n$, we have 
$S_{nt} \vert_{\mathcal{C}} = Id \vert_{\mathcal{C}}$, and so 
$ D S_{nt}(\hat{q}) . v = v $ for any $v \in T_{\hat{q}} \mathcal{C}$. 
This implies that $ \Vert D S^{nT}(\hat{q}).v \Vert \leq \Vert v \Vert $ for any $T > 0$ and $n \in \mathbb{Z}$. 

We consider now $(\varphi_{0}, y_{0})$ and $(\varphi(t), y(t))$ the associated solution. 
The estimates from the lemma \ref{th:lem4} above imply 
 $\vert \varphi(t) - \varphi_{0} \vert \leq 2{\gamma(\rho)} \frac{1}{\beta} \Vert y_{0} \Vert, \;\;\;
 \text{and} \;\;\; \Vert y(t) \Vert \leq 2 \Vert y_{0} \Vert e^{- \beta t} ,$
where $\rho = 8 (\vert \varphi_{0} \vert + \Vert y_{0} \Vert)$.\\
For $v \in R(L_{\hat{q}})=N^{s}_{\hat{q}}$, we have $v = (0, y_{v}) = (0, v)$ in the new coordinates, 
and if we consider $q(0) = (\varphi_{0}, y_{0}) = \hat{q} + \epsilon v$, we find 
\begin{equation} \vert \varphi(t) \vert + \Vert y(t) \Vert \leq \left [ \frac{ 2 \gamma(\rho)}{ \beta} 
+ 2{e^{- \beta t}} \right ] \Vert y_{0} \Vert. \end{equation}

With $\Vert y_{0} \Vert = \frac{1}{8} \rho = \epsilon \Vert v \Vert$, we have 
$\gamma(\rho) \underset{\epsilon \rightarrow 0}{\longrightarrow} 0$, and so there is a 
$\epsilon_{1}$ such that for any $\epsilon \leq \epsilon_{1}$ we have 
$ \frac{ 2 \gamma(\rho)}{ \beta} \leq \frac{1}{10} $. 
And choosing $T$ such that $e^{- \beta T} \leq \frac{1}{20}$ gives 
\begin{equation} \Vert S^{T}(\hat{q} + \epsilon v) - S^{T} (\hat{q}) \Vert \leq
 \left [ \frac{1}{10} + 2 \frac{1}{20} \right ] \epsilon \Vert v \Vert \leq \frac{\epsilon}{5} \Vert v \Vert, \end{equation}
and $\Vert D S^{T}(\hat{q}) . v \Vert \leq \frac{1}{5} \Vert v \Vert.  $ 

In the same way, for any $n \geq 1$, there is a $ \epsilon_{n} >0$ such that 
for all $\epsilon < \epsilon_{n} < \epsilon_{1}$, we have
$\frac{ 2 \gamma(\rho)}{ \beta} \leq \frac{1}{10^{n}}$, and then  
\begin{equation} \Vert S^{nT}(\hat{q} + \epsilon v) - S^{nT} \hat{q} \Vert \leq 
\left [ \frac{1}{10^{n}} + 2 \frac{1}{20^{n}} \right ] \epsilon \Vert v \Vert \leq \frac{2 \epsilon}{10^{n}} \Vert v \Vert \end{equation}
for any $v \in R(L_{\hat{q}})= N^{s}_{\hat{q}}$, 
and this implies 
 $ \Vert D S^{nT}(\hat{q}) . v \Vert \leq \frac{2}{10^{n}} \Vert v \Vert.  $
This completes the proof of normal hyperbolicity of $\mathcal{C}$. 

\end{document}